\documentclass[sigconf]{acmart}

\AtBeginDocument{%
  }

\setcopyright{acmlicensed}
\copyrightyear{2018}
\acmYear{2018}
\acmDOI{XXXXXXX.XXXXXXX}
\acmConference[Conference acronym 'XX]{Make sure to enter the correct
conference title from your rights confirmation email}{June 03--05, 2018}{Woodstock, NY}
\acmISBN{978-1-4503-XXXX-X/2018/06}
\usepackage[ruled,vlined,linesnumbered]{algorithm2e}

\makeatletter
\newcommand{\algoherestart}{\@twocolumnfalse}
\newcommand{\algoherestop}{\@twocolumntrue}
\makeatother
\usepackage{bm}
\usepackage{subcaption}
\usepackage[table]{xcolor}
\usepackage{multirow}
\usepackage{siunitx}
\newcommand{\best}[1]{\textbf{#1}}
\newcommand{\second}[1]{\underline{#1}}
\usepackage{float}
\usepackage{flafter}
\usepackage{placeins}
\usepackage[most]{tcolorbox}

\newtcolorbox{promptbox}[1]{
  enhanced,
  colback=gray!8,
  colframe=black!70,
  colbacktitle=black!75,
  coltitle=white,
  fonttitle=\bfseries,
  title=#1,
  boxrule=0.6pt,
  arc=2pt,
  left=8pt, right=8pt, top=6pt, bottom=6pt,
  toptitle=4pt, bottomtitle=4pt,
  float=htbp,
  floatplacement=htbp,
  before skip=8pt, after skip=8pt,
}
\newtheorem{theorem}{Theorem}[section]
\newtheorem{lemma}{Lemma}[section]

\newtheorem{proposition}{Proposition}[section]
\newtheorem{corollary}{Corollary}[section]

\SetKwInput{KwIn}{Input}
\SetKwInput{KwOut}{Output}
\SetKw{KwRet}{Return}

\begin{document}

\title[PURA: Provably Unbiased and Robust Multi-Bit Text Attribution]{PURA: Provably Unbiased and Robust Multi-Bit Watermarking for AI-Generated Text Attribution}


\author{Yaofei Wang}
\email{wyf@hfut.edu.cn}
\affiliation{%
  \institution{Hefei University of Technology}
  \city{Hefei}
  \country{China}
}

\author{Jinyang Guo}
\email{2025170798@mail.hfut.edu.cn}
\affiliation{%
  \institution{Hefei University of Technology}
  \city{Hefei}
  \country{China}
}

\author{Shuchao Du}
\email{dushuchao25@mails.ucas.ac.cn}
\affiliation{%
  \institution{University of Chinese Academy of Sciences}
  \city{Beijing}
  \country{China}
}

\author{Chao Wang}
\email{chaowang0708@mail.ustc.edu.cn}
\affiliation{%
  \institution{University of Science and Technology of China}
  \city{Hefei}
  \country{China}
}

\author{Qiyi Yao}
\email{qiyi.yao@bilkent.edu.tr}
\affiliation{%
  \institution{Bilkent University}
  \city{Ankara}
  \country{Türkiye}
}

\author{Donghui Hu}
\email{hudh@hfut.edu.cn}
\affiliation{%
  \institution{Hefei University of Technology}
  \city{Hefei}
  \country{China}
}

\author{Weiming Zhang}
\email{zhangwm@ustc.edu.cn}
\affiliation{%
  \institution{University of Science and Technology of China}
  \city{Hefei}
  \country{China}
}

\author{Nenghai Yu}
\email{ynh@ustc.edu.cn}
\affiliation{%
  \institution{University of Science and Technology of China}
  \city{Hefei}
  \country{China}
}

\author{Kejiang Chen}
\authornote{Corresponding author.}
\email{chenkj@ustc.edu.cn}
\affiliation{%
  \institution{University of Science and Technology of China}
  \city{Hefei}
  \country{China}
}

\renewcommand{\shortauthors}{Wang et al.}

\begin{abstract}
Fine-grained attribution of AI-generated text is becoming increasingly important for accountability and auditing, yet existing multi-bit watermarking methods still struggle to simultaneously preserve the base generation distribution, support high-capacity payloads, and remain recoverable after editing. We present PURA, a provably unbiased and robust multi-bit watermarking method for text attribution. Instead of perturbing token probabilities directly, PURA embeds payloads in the latent sampling space via keyed inverse transform sampling, and recovers them by treating observed tokens as soft interval evidence and aggregating such evidence across the sequence. This design preserves the base generation distribution exactly while substantially improving recovery stability under post-editing and channel perturbations. Building on this recovery paradigm, we further develop a unified robustness analysis and show that, under bounded attack strength, the per-bit error probability decays exponentially with sequence length. Extensive experiments show that PURA substantially outperforms existing unbiased baselines in the high-payload regime. For example, when embedding 36 bits in 200 tokens, PURA achieves a 91.7\% message match rate, more than three times that of the strongest unbiased baseline, while preserving text quality and remaining statistically close to unwatermarked text, and incurring only millisecond-level verification overhead. Our code is available at \url{https://github.com/anshigaoa/PURA}.
\end{abstract}



\maketitle

\section{Introduction}

Text generated by large language models (LLMs) is increasingly difficult to distinguish from human writing~\cite{achiam2023gpt4,
touvron2023llama}, and provenance tracing has correspondingly
shifted from a research problem to a deployment requirement. In
many settings, it is no longer sufficient to decide whether a
passage was generated by a model. One also needs finer-grained
attribution information, such as the serving instance, the
requesting user, the generation time, or the model version. This
provenance underpins accountability, internal auditing, and
compliance: privately deployed models can be misused by insiders,
generated content can be reposted or anonymized across platforms,
and regulations increasingly push AI-generated content toward
traceability and attribution~\cite{smuha2025regulation,
biden2023executive}. Existing zero-bit watermarking schemes only
answer whether a text comes from a watermarked
model~\cite{kirchenbauer2023KGW,zhao2024provable,kuditipudi2023its,
ICLR2024_c5b00c5b}, whereas multi-bit watermarking embeds structured
attribution information during generation, moving the question
from whether a text is AI-generated to who generated it, when, and
under which model version.

Despite recent progress, existing multi-bit watermarking methods
still struggle to satisfy three objectives simultaneously. The first
is distribution preservation: a watermark should not alter the base
model's generation distribution, since sampling bias degrades text
quality and downstream utility. The second is high capacity: the
recoverable payload per unit of text should be large enough to
carry meaningful attribution information rather than only a short
marker. The third is editing robustness: the embedded payload
should remain recoverable after substitution, deletion, local
rewriting, or semantic paraphrasing. Biased schemes typically
improve capacity or recovery strength by directly perturbing token
probabilities, but introduce a clear trade-off among capacity,
robustness, and text quality~\cite{yoo2023MPAC,DBLP:conf/iclr/WangYC0LM0024,qu2024via_code}.
Unbiased schemes preserve the generation distribution more
faithfully, but often degrade sharply under high payloads and
post-editing attacks~\cite{fernandez2023cycleshift,feng2025bimark,jiang2025stealthink}.
More importantly, formal robustness claims in prior work largely
target generic editing attacks, and meaningful bounds against
attackers who understand the watermarking mechanism remain scarce.

To address these challenges, we propose PURA, a provably unbiased
and robust multi-bit watermarking framework for AI-generated text
attribution. PURA rethinks multi-bit attribution from both the
embedding side and the extraction side. On the embedding side, it
follows the sampling-level distortion-free watermarking line~\cite{kuditipudi2023its},
but moves beyond ITS-style zero-bit detection by encoding
structured multi-bit payloads in the latent sampling space,
adaptively allocating embedding capacity according to local
entropy, and using a locality-preserving coding design to mitigate
bit diffusion under local perturbations. On the extraction side,
PURA departs from hard local decoding and instead treats tokens as
soft interval evidence aggregated across the sequence into per-bit
statistics. This design preserves the base generation distribution
while allowing weak local evidence to accumulate across positions,
leading to substantially more stable recovery under text editing
and model mismatch.

Building on these per-bit aggregated statistics, we develop a
robustness analysis framework that covers two realistic attack
regimes: generic editing attacks that perturb the text through
substitution, deletion, insertion, or rewriting, and directed
attacks in which the adversary understands the watermarking
mechanism and leverages external models or proxy pipelines. Unlike
prior work that bounds robustness at the edit-distance level and
only against generic editing, we analyze recovery directly at the
soft-evidence statistic level and decompose attack effects into a
structural synchronization channel and a directional evidence
channel, bringing both attack regimes into a single framework.
PURA thus yields a per-bit error bound that decays exponentially
with sequence length and characterizes how structural
desynchronization and directional evidence drift jointly determine
recovery performance. The same embedding and extraction pipeline
supports both symmetric verification, where the verifier accesses
the source model, and asymmetric verification, where the verifier
only uses a compatible proxy model from the same family, which is
important when the source model is proprietary or cannot be exposed
directly.

We evaluate PURA on three datasets, C4, Essays, and OpenGen~\cite{rael_exploring_nodate,ivypanda2024,merity_pointer_2016},
across four model families, Llama-3, Gemma-2, BLOOM, and Pythia~\cite{grattafiori_llama_2024,gemma2_2024,bloom_2022,pythia_icml_2023},
under a wide range of editing and rewriting attacks. Even in the
asymmetric setting, PURA maintains over 90\% message match rate
and over 99\% bit accuracy across the 12--36 bits per 200 tokens
regime, substantially outperforming biased baselines in the
high-payload range and far exceeding unbiased baselines. PURA
also remains close to unwatermarked text in both quality and
statistical stealth, and asymmetric verification requires only
about 0.05 seconds per sample.

Our main contributions are as follows:
\begin{itemize}
\item We propose latent-sector multi-bit anchoring, a
distribution-preserving embedding mechanism that encodes payload
bits as context-scheduled, Gray-coded rotations of the latent
sampling point, supporting high-capacity attribution without
altering the base generation distribution.

\item We propose soft-evidence extraction for robust
multi-bit recovery. Each observed token is treated as an interval
in the latent space, converted to sector-level compatibility
weights, aggregated into per-bit log-likelihood ratios (LLRs) across the sequence, and used
to guide lightweight confidence-based repair.

\item To the best of our knowledge, this is the first robustness analysis framework for multi-bit watermarking that covers
mechanism-aware directed attacks, yielding a per-bit error bound
that decays exponentially with sequence length.
\end{itemize}
\section{Related Work}

\subsection{LLM Watermarking}

Research on watermarking for large language models begins with zero-bit schemes, whose goal is to decide whether a text was produced by a watermarked model. KGW~\cite{kirchenbauer2023KGW} partitions the vocabulary into green and red lists through an n-gram context hash and adds a logit bias to green tokens, with detection based on a statistical test over green-token counts. Unigram~\cite{zhao2024provable} replaces the context-dependent partition with a global green list fixed by the master key, trading flexibility for stronger robustness to substitution and paraphrase. Both schemes modify the output distribution and therefore expose a tension among detection strength, text quality, and statistical stealth.

Subsequent work explores distribution-preserving watermarking under key randomness. One line still reweights per-step token probabilities, but constructs the reweighting so that the expectation over keys matches the original distribution, including $\gamma$-reweight, DiPmark, STA, MCmark, and SynthID-Text~\cite{ICLR2024_c5b00c5b,wu2023dipmark,mao2024STA,mcmark2025Chen,dathathri2024synthid}. Another line leaves token probabilities unchanged and instead couples key randomness directly with the sampling procedure, including Aaronson's Gumbel-max scheme and the ITS/EXP schemes of Kuditipudi et al.~\cite{aaronson2023openai,kuditipudi2023its}, while PURA's embedding stage builds on this sampling-level line. All zero-bit schemes, however, only decide watermark presence and carry no structured provenance information.

Multi-bit watermarking instead embeds short payloads during generation, such as user, timestamp, or version information. Existing methods can likewise be divided according to whether they alter the output distribution. Biased schemes typically extend the KGW-style logit-bias framework with message encoding. MPAC~\cite{yoo2023MPAC} routes each token to a payload position through a context hash and biases the color list matching the target symbol. CTWL~\cite{DBLP:conf/iclr/WangYC0LM0024} uses a proxy language model to select probability-balanced token subsets before applying the bias. RS-BH~\cite{qu2024via_code} packs bits into segments and combines message-dependent hashing with Reed--Solomon correction. This family usually achieves strong recovery at high payloads, at the cost of degraded quality and statistical stealth as payload increases. Unbiased schemes preserve the distribution in expectation over the key. CycleShift~\cite{fernandez2023cycleshift} assigns each message a cyclically shifted secret vector, BiMark~\cite{feng2025bimark} uses XOR-masked multi-layer bit-flip reweighting, and StealthInk~\cite{jiang2025stealthink} suppresses a message-aligned interval on a permuted vocabulary with compensatory reweighting on a symmetric interval. Compared with biased methods, this family better preserves quality and stealth, but both capacity and extraction accuracy remain limited under high payloads and post-editing.

Concurrent works MC$^2$Mark~\cite{cui2026mc2markdistortionfreemultibitwatermarking} and ArcMark~\cite{gilani2026arcmarkdistortionfreemultibytellm} also study distortion-free multi-bit watermarking with different embedding and decoding constructions.


\subsection{Attacks on LLM Watermarks}

Early robustness evaluations mostly rely on mechanism-agnostic perturbations, including token- or character-level insertion, deletion, and substitution, as well as paraphrase-style rewriting such as DIPPER~\cite{krishna2023paraphrasing}. The effectiveness of such attacks is largely governed by edit distance or semantic drift rather than by the structure of the watermarking algorithm itself.

Mechanism-aware attacks explicitly exploit the structure of the target watermark, but differ in the assumptions they rely on. Attacks against biased watermarks mainly exploit persistent green-list bias. Watermark Stealing~\cite{pmlr-v235-jovanovic24a} queries a watermarked model at scale to approximate the bias structure of KGW and Unigram for spoofing and scrubbing, and DE-MARK~\cite{chen2025demark} recovers the prefix length, bias strength, and green/red partition of $n$-gram watermarks through probing. Both require a token-level bias that persists across samples and therefore do not apply to distribution-preserving watermarks.

Attacks against unbiased watermarks often rely on stronger structural assumptions. Reynolds et al.~\cite{reynolds2025breaking} attack the distortion-free ITS scheme of Kuditipudi et al.~\cite{kuditipudi2023its} by recovering its secret permutation and key sequence through comparison sorting, under the assumption of a globally fixed permutation and a position-wise cyclic key sequence independent of context. Zhang et al.~\cite{zhang2026character} exploit tokenizer-sensitive character perturbations to spread the effect of a single character edit across multiple neighboring token positions, with the strongest variant additionally relying on detector-guided search.

A third class assumes neither biased sampling structure nor fixed secret parameters or detector feedback. Diaa et al.~\cite{diaa2025optimizing} and Huang et al.~\cite{huang2025RLCracker} fine-tune open-source paraphrasers through preference optimization and reinforcement learning so that the rewritten text moves away from the watermarked distribution while preserving semantics. SIRA~\cite{pmlr-v267-cheng25c} performs targeted rewriting through masking and refilling of high-entropy positions, and the Smoothing Attack~\cite{chang-etal-2025-watermark} selectively resamples low-confidence positions. These attacks are closer to the threat model considered in PURA, since they reflect the setting in which an attacker has no secret key or detector feedback but can still leverage external models and knowledge of the watermarking mechanism.

\subsection{Linguistic Steganography}
Linguistic steganography conceals messages within natural language, fundamentally prioritizing stealth for covert communication over survivability for provenance authentication. Provably secure steganography methods ~\cite{Ziegler2019Neural,Kaptchuk2021Meteor,iMEC,Ding2023Discop,Wang2025SparSamp,wang2026ANStega} in symmetric settings achieve remarkable capacity and security. However, this comes at the cost of strict synchronization: the decoder must hold the exact language model and flawlessly reproduce the probability distribution step-by-step. In realistic adversarial settings, even trivial token-level edits perturb the shared context, causing cascading desynchronization and catastrophic recovery failures.


To alleviate the heavy burden of decoder-side synchronization, recent work has explored asymmetric constructions. For example, Bai et al.~\cite{bai_provably_2025} enable the decoder to infer a distribution permutation directly from the carrier, and Disreo~\cite{11329500} extracts bits from token positions within a permuted vocabulary without accessing the original language model. While successfully relaxing the hardware and model access limits, these methods incur a precipitous drop in embedding capacity---often requiring tens of tokens to convey a single bit---and remain highly sensitive to local edits. Ultimately, because steganography (whether symmetric or asymmetric) is inherently designed for secrecy over a cooperative channel rather than survivability against an active adversary, existing steganographic schemes fail to meet the joint high-capacity and robust attribution requirements of fine-grained LLM auditing.

\section{Problem Formulation and Preliminaries}
\label{sec:problem_setting}

\subsection{Multi-bit Watermarking and Verification Settings}
\label{subsec:problem_formulation}

We study multi-bit watermarking for autoregressive language models.
Let $\mathcal{M}_{\mathrm{src}}$ denote the source model used during
generation. Given a user prompt, the encoder generates a text
sequence $\mathbf{x} = (x_1, \dots, x_L)$ while embedding a binary
payload $m \in \{0,1\}^{L_m}$ into the generation process.

At verification time, the verifier observes either the original text
$\mathbf{x}$ or an edited version $\mathbf{x}'$ and outputs a
recovered payload $\hat{m} \in \{0,1\}^{L_m}$. We use two evaluation
metrics: \emph{bit accuracy}, the fraction of payload bits recovered
correctly, and \emph{match rate}, the probability that the entire
payload is recovered without error.

We consider two verification settings that differ in the verifier's
access to the source model. In the \emph{symmetric setting}, the
verifier uses the same model $\mathcal{M}_{\mathrm{src}}$ that
generated the text, which yields the strongest recovery signal and
fits internal auditing scenarios. In the \emph{proxy-based setting},
the verifier cannot access $\mathcal{M}_{\mathrm{src}}$ and instead
uses a compatible open proxy model $\mathcal{M}_{\mathrm{prx}}$ from
the same family, which approximates the source-model next-token
distribution without exposing the source weights and is more
realistic when the source model is proprietary. Both settings share
the same embedding and recovery pipeline. The additional challenge
in the proxy-based setting is source--verifier distribution mismatch,
which motivates the soft-evidence aggregation design introduced in
Section~\ref{sec:method}. We use $\mathcal{M}_{\mathrm{ver}}$ to
denote the verifier-side model, namely $\mathcal{M}_{\mathrm{src}}$
in the symmetric setting and $\mathcal{M}_{\mathrm{prx}}$ in the
proxy-based setting.

\subsection{Autoregressive Generation and ITS}
\label{subsec:its_background}

At each step $t \in \{1,\dots,L\}$, the language model produces a
next-token distribution $\pi_t$ over the vocabulary $\mathcal{V}$
conditioned on the preceding context and samples the current output
token from $\pi_t$.

A standard way to sample from a discrete $\pi_t$ is inverse
transform sampling (ITS). Fix an ordering
$(v_1, \dots, v_{|\mathcal{V}|})$ of $\mathcal{V}$ and define the
cumulative distribution function (CDF) by
\begin{equation}
  F_t(r) = \sum_{k=1}^{r} \pi_t(v_k), \qquad F_t(0) = 0.
  \label{eq:cdf}
\end{equation}
Given a uniform variable $u \sim U[0,1)$, ITS returns
\begin{equation}
  r = \min\{k \mid F_t(k) > u\}, \qquad x_t = v_r.
  \label{eq:its}
\end{equation}
This samples exactly from $\pi_t$, and each sampled token $x_t$
corresponds to an interval on $[0,1)$:
\begin{equation}
  I_t(x_t) = \bigl[F_t(r-1),\, F_t(r)\bigr).
  \label{eq:token_interval}
\end{equation}

PURA further relies on the rotational invariance of the uniform
distribution on $[0,1)$: for any fixed offset $\Delta \in [0,1)$,
\begin{equation}
  u \sim U[0,1) \;\Longrightarrow\; u \oplus \Delta \sim U[0,1),
  \label{eq:rotation}
\end{equation}
where $\oplus$ denotes addition modulo $1$. Replacing $u$ with
$u \oplus \Delta$ in Equation~\eqref{eq:its} therefore leaves the
induced token distribution unchanged, so payload information can be
encoded as a phase offset in the latent sampling space without
altering the base generation distribution.

\subsection{Threat Model}
\label{subsec:threat_model}

We consider two attacker models ordered by capability. Both follow
the spirit of Kerckhoffs's principle: the attacker may know the
complete watermarking algorithm but cannot obtain the key
$\mathsf{sk}$ and cannot brute-force the key space in polynomial
time. Throughout the attack, the attacker receives no detection
feedback.

\noindent\textbf{TM0: non-adaptive attacker.}
The attacker observes a watermarked text $\mathbf{x}$ and applies
generic post-processing edits to obtain $\mathbf{x}'$, including
token- and character-level insertion, deletion, substitution, and
local rewriting. These edits are applied independently of the
watermark structure, and character-level perturbations are
absorbed through their effect on tokenization.

\noindent\textbf{TM1: adaptive attacker.}
In addition to the capabilities of TM0, the attacker has full access
to the public specification of the watermarking mechanism and may
leverage auxiliary resources such as open-source language models,
public collections of watermarked text, and proxy watermarking
pipelines instantiated from that specification. Using this
algorithmic knowledge, the attacker mounts more targeted attacks
that, at comparable semantic cost, degrade recovery more
efficiently than mechanism-agnostic edits.

Under both threat models, the verifier is expected to retain
non-trivial per-bit recovery as long as enough synchronized
evidence survives editing. Section~\ref{sec:theoretical_analysis}
formalizes this condition.
\section{Method}
\label{sec:method}

PURA is a distribution-preserving multi-bit attribution scheme
that uses inverse transform sampling (ITS) as an unbiased sampling
primitive. The core idea is to encode payload bits by shifting the
latent sampling anchor while keeping the induced next-token
distribution unchanged. At verification time, PURA does not decode
each token into a hard bit pattern. Each observed token instead
defines an interval in the latent sampling space, which PURA
converts into soft bit evidence and aggregates across the sequence.

The method has four components. Context-gated payload scheduling
decides whether a step is usable, how many bits it can carry, and
which payload positions it reads. Gray-coded latent anchoring then
maps the selected bits to a sector of the unit interval. Keyed ITS
over a permuted vocabulary samples the token without changing the
base distribution. Finally, soft-evidence extraction integrates
interval evidence, aggregates per-bit LLRs, and applies
lightweight reliability-guided repair. The repair step uses a
single-parity-check (SPC) code.
Each component addresses a distinct design constraint.
Context-gated scheduling avoids assigning many bits to unreliable
low-entropy steps, while Gray-coded latent anchoring localizes errors
caused by sector drift. Keyed ITS preserves the generation distribution,
and soft-evidence aggregation with parity repair supports recovery after
editing.
\subsection{Embedding: Context-Scheduled Latent Anchoring}
\label{subsec:embedding}

Let $m \in \{0,1\}^{L_m}$ be a binary payload. Before embedding, we
partition $m$ into $B_{\mathrm{blk}}$ blocks
$m^{(i)} \in \{0,1\}^{\ell}$ of length $\ell$, and apply a
lightweight blockwise parity encoding to obtain the encoded payload
$c$:
\begin{equation}
  p_i = \bigoplus_{j=1}^{\ell} m^{(i)}_j,
  \label{eq:parity}
\end{equation}
\begin{equation}
  c = \bigl(m^{(1)} \| p_1 \| \cdots \|
      m^{(B_{\mathrm{blk}})} \|
      p_{B_{\mathrm{blk}}}\bigr)
  \;\in\; \{0,1\}^{|c|},
  \quad |c| = B_{\mathrm{blk}}(\ell+1).
  \label{eq:encoded_payload}
\end{equation}
Figure~\ref{fig:embedding} summarizes the resulting single-step
embedding pipeline.

\noindent\textbf{Context-Gated Payload Scheduling.}
At each step, PURA assigns an embedding arity $a_t$ based on the
entropy of the next-token distribution $\pi_t$. Entropy controls
how finely the unit interval can be sectorized: high-entropy steps
yield narrow token intervals that resolve fine sectors, while
low-entropy steps produce wide intervals that blur sector
boundaries and erase the embedded bits. Low-entropy steps are
therefore skipped or assigned fewer bits, while high-entropy steps
support finer sectorization. In principle $a_t$ can take arbitrary
non-negative values, but larger arities shrink sectors faster than
typical entropy can sustain under editing, so we cap
$a_t \in \{0,1,2,3\}$ throughout this paper. Appendix
~\ref{app:gating_ablation} confirms this schedule.

PURA uses two keyed and domain-separated hash functions:
$\mathcal{H}_{\mathrm{seed}}$ for seed generation and
$\mathcal{H}_{\mathrm{idx}}$ for payload indexing. Both produce
$\lambda$-bit outputs, which we map to integers and to $[0,1)$ by
\begin{equation}
  \mathrm{Int}(h)\in\{0,\dots,2^\lambda-1\},
  \qquad
  \mathrm{Unif}(h)=\mathrm{Int}(h)/2^\lambda \in [0,1).
  \label{eq:hash_int_unif}
\end{equation}
All keyed derivations use distinct public domain prefixes for their
respective purposes.
To avoid reusing the same latent randomness under repeated local
contexts, PURA skips a step if the seed hash of the current context
has appeared earlier in the same session:
\begin{equation}
  a_t = 0 \quad \text{if} \quad
  \mathcal{H}_{\mathrm{seed}}(x_{t-W:t-1};\mathsf{sk})
  \in \mathsf{Hist},
  \label{eq:uniqueness}
\end{equation}
where $W$ is the context window size and
$\mathsf{Hist}$ stores previous seed hashes.

For each usable step, PURA maps the current context to a cyclic
payload index:
\begin{equation}
  \mathrm{idx}(t) =
  \mathrm{Int}\!\bigl(
    \mathcal{H}_{\mathrm{idx}}(x_{t-W:t-1};\mathsf{sk})
  \bigr) \bmod |c|.
  \label{eq:idx}
\end{equation}
Starting from $\mathrm{idx}(t)$, PURA reads $a_t$ consecutive bits from $c$
in cyclic order. The selected bit sequence is denoted by
\[
b_t=\bigl((b_t)_0,\ldots,(b_t)_{a_t-1}\bigr),
\]
where
\begin{equation}
  (b_t)_j = c_{(\mathrm{idx}(t)+j) \bmod |c|},
  \qquad j \in \{0, \dots, a_t - 1\}.
  \label{eq:payload_bits}
\end{equation}
This scheduling avoids sequential assignment and spreads evidence
more evenly over payload positions.

\noindent\textbf{Gray-Coded Latent Anchoring.}
Given the selected bits $b_t$, PURA partitions the unit interval into
$K_t = 2^{a_t}$ equal-width sectors and labels them with a Gray code.
Let $G_{a_t}:\{0,\dots,K_t{-}1\}\to\{0,1\}^{a_t}$ denote the Gray map
of arity $a_t$, and write $G_{a_t}^{-1}$ for its inverse. The
selected bits determine the sector index:
\begin{equation}
  S_t = G_{a_t}^{-1}(b_t) \in \{0,\dots,K_t-1\}.
  \label{eq:gray_sector}
\end{equation}
Gray coding localizes sector drift. Under ordinary binary indexing,
two adjacent sectors may differ in several bits, for example
$011 \to 100$ when $a_t = 3$. Under Gray coding, adjacent sectors
differ in exactly one bit, so a small latent shift is more likely to
create a sparse bit error rather than a multi-bit error.

PURA derives a pseudo-random base phase from the current context:
\begin{equation}
  U_t = \mathrm{Unif}\!\bigl(
    \mathcal{H}_{\mathrm{seed}}(x_{t-W:t-1};\mathsf{sk})
  \bigr) \in [0,1),
  \label{eq:base_seed}
\end{equation}
and rotates it to the center of the selected sector:
\begin{equation}
  A_t = U_t \oplus \frac{S_t + 0.5}{K_t},
  \qquad a_t > 0,
  \label{eq:anchor}
\end{equation}
where $\oplus$ denotes addition modulo $1$. The payload therefore
selects a latent sampling anchor, rather than modifying token
probabilities.

\noindent\textbf{Distribution-Preserving Keyed Sampling.}
PURA samples from the model distribution by applying ITS in a
key-derived permuted vocabulary order. Let $\Pi$ be a key-derived
permutation over $\mathcal{V}$, and let the permuted vocabulary be
$\tilde{\mathcal{V}}=(\tilde{v}_1, \dots,\tilde{v}_{|\mathcal{V}|})$.
For any token $v\in\mathcal{V}$, write $\mathrm{rank}_\Pi(v)$ for its
position in $\tilde{\mathcal{V}}$, i.e., the unique $r$ with
$\tilde{v}_r = v$. The corresponding secret-space CDF is
\begin{equation}
  \tilde{F}_t(r) = \sum_{k=1}^{r} \pi_t(\tilde{v}_k),
  \qquad \tilde{F}_t(0) = 0.
  \label{eq:permuted_cdf}
\end{equation}
The permutation keeps the token-to-latent correspondence secret,
which prevents an attacker with a compatible public model from
directly using the public distribution and a fixed latent ordering to
choose targeted replacements.

When $a_t = 0$, PURA samples directly from $\pi_t$ without payload
anchoring. Otherwise, it samples by
\begin{equation}
  r_t = \min\{k \mid \tilde{F}_t(k) > A_t\},
  \qquad x_t = \tilde{v}_{r_t}.
  \label{eq:secret_sampling}
\end{equation}
Since $U_t$ is pseudo-random uniform under a fresh context and
modular rotation preserves uniformity, $A_t$ remains uniform for any
fixed payload-dependent sector offset. Therefore, ITS with anchor
$A_t$ samples exactly from $\pi_t$. Formal unbiasedness is given in
Theorem~\ref{thm:comp_unbiasedness}.

\begin{figure}[t]
  \centering
  \includegraphics[width=\linewidth]{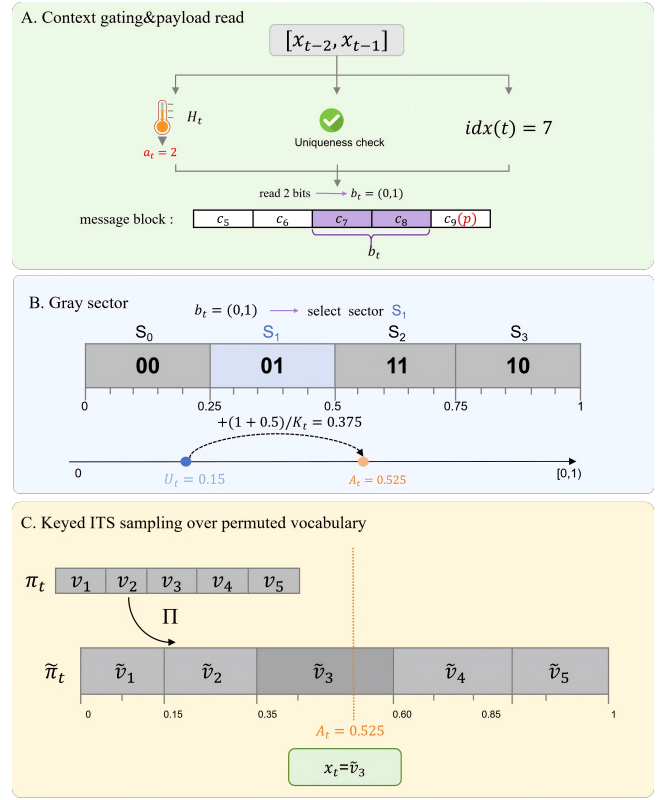}
  \caption{Single-step embedding in PURA.
\textbf{(A)} The context determines whether the step is usable,
its bit-width, and the payload positions to read.
\textbf{(B)} The selected bits index a Gray-coded sector that
rotates the base phase $U_t$ into a latent anchor $A_t$.
\textbf{(C)} Keyed ITS over the permuted vocabulary maps $A_t$ to
a token. The payload acts only on the latent anchor, leaving the
induced next-token distribution unchanged.}
  \Description{A three-stage embedding pipeline in which context gating selects payload bits, Gray coding rotates a uniform phase into a sector anchor, and keyed inverse-transform sampling over a permuted vocabulary emits the next token.}
  \label{fig:embedding}
\end{figure}

\subsection{Extraction: Soft Evidence Aggregation}
\label{subsec:extraction}

Let $\mathcal{M}_{\mathrm{ver}}$ denote the verifier-side model,
namely $\mathcal{M}_{\mathrm{src}}$ in the symmetric setting and
$\mathcal{M}_{\mathrm{prx}}$ in the proxy-based setting. Let
$\pi_t^{\mathrm{ver}}$ be the verifier-side next-token distribution
at step $t$, and let $\tilde{F}_t^{\mathrm{ver}}$ be the secret-space
CDF obtained by applying the same permutation $\Pi$ to
$\pi_t^{\mathrm{ver}}$ as in Equation~\eqref{eq:permuted_cdf}. Given
an observed sequence, which may be edited, the verifier recomputes
the same context gate, payload index, secret vocabulary order, and
verifier-side phase
\[
  U_t^{\mathrm{ver}} =
  \mathrm{Unif}\!\bigl(
    \mathcal{H}_{\mathrm{seed}}(x_{t-W:t-1};\mathsf{sk})
  \bigr).
\]
When the context window survives editing, $U_t^{\mathrm{ver}}$ agrees
with the embedding-side phase. Otherwise, the step is desynchronized
and contributes little aligned signal in aggregation.

Unlike recovery schemes based on local hard decisions, PURA does not
assign each observed token to a single sector or bit pattern. Each
observed token defines an interval in the latent sampling space.
PURA converts this interval into sector-level compatibility weights,
projects these weights into bit-level LLRs, and aggregates the LLRs
over payload positions. Figure~\ref{fig:extraction} illustrates this
pipeline.
\begin{figure*}[!tp]
  \centering
  \includegraphics[width=0.95\linewidth]{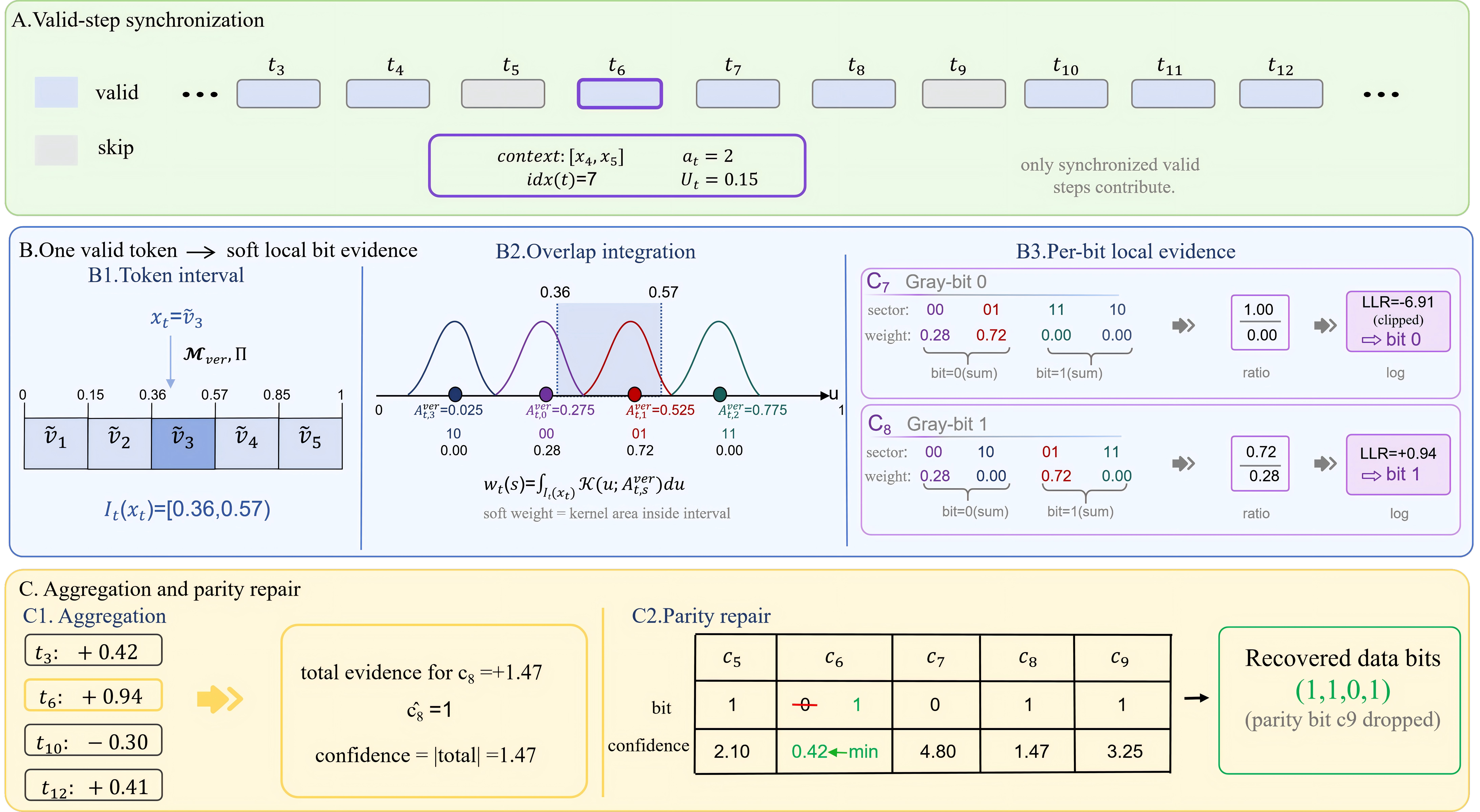}
  \caption{PURA extraction by soft evidence aggregation.
\textbf{(A)} The verifier retains synchronized valid steps and
recovers each payload index from context.
\textbf{(B)} The token's interval in the permuted secret
vocabulary is integrated against each sector anchor through a
symmetric kernel. Sector weights are then projected onto Gray-bit
coordinates to form local LLRs.
\textbf{(C)} LLRs are aggregated per payload position, with the sign
giving the bit estimate and the magnitude guiding parity-based
repair.}
  \Description{A three-stage extraction pipeline in which synchronized valid tokens define intervals, kernel overlap produces sector and bit-level evidence, and aggregated LLRs with parity repair recover the payload.}
  \label{fig:extraction}
\end{figure*}
\noindent\textbf{Interval-Based Sector Weighting.}
For each observed token $x_t$, the verifier first locates its
probability interval in the verifier-side secret vocabulary space:
\begin{equation}
  I_t(x_t) = \Bigl[
    \tilde{F}_t^{\mathrm{ver}}\!\bigl(
      \mathrm{rank}_\Pi(x_t)-1\bigr),\;
    \tilde{F}_t^{\mathrm{ver}}\!\bigl(
      \mathrm{rank}_\Pi(x_t)\bigr)
  \Bigr).
  \label{eq:obs_interval}
\end{equation}
Thus, an observed token corresponds to an interval on $[0,1)$, not
to a single latent point.

For each candidate sector $s\in\{0,\dots,K_t{-}1\}$, let
$A_{t,s}^{\mathrm{ver}} = U_t^{\mathrm{ver}} \oplus (s+0.5)/K_t$
be its candidate anchor. PURA measures how compatible the observed
interval is with this anchor by integrating a symmetric kernel over
the interval:
\begin{equation}
  w_t(s) =
  \int_{I_t(x_t)}
  \mathcal{K}(u;\, A_{t,s}^{\mathrm{ver}})\, du,
  \label{eq:sector_weight}
\end{equation}
where $\mathcal{K}$ is a symmetric kernel. The weight $w_t(s)$ is a
soft compatibility score: one token can support several nearby
sectors with different strengths. Since the kernel form is not
central to the recovery paradigm, we give the unified formulation
here and defer kernel definitions and ablations to
Appendix~\ref{app:kernels} and Section~\ref{subsec:analysis}.

\noindent\textbf{Local Bit Evidence.}
The sector weights are then projected to Gray-bit coordinates. For
each coordinate $j \in \{0, \dots, a_t{-}1\}$, define
\[
\mathcal{S}^{(j)}_b =
\bigl\{s \in\{0,\dots,K_t{-}1\}\;\big|\; [G_{a_t}(s)]_j = b\bigr\},
\quad b\in\{0,1\}.
\]
Let $p_t^{(j)}\in(0,1)$ denote the bit-$1$ posterior implied by the
sector weights, namely
\[
p_t^{(j)}
=
\frac{\sum_{s\in\mathcal{S}^{(j)}_1} w_t(s)}
     {\sum_{s\in\{0,\dots,K_t-1\}} w_t(s)}.
\]
For numerical stability, we clip this posterior to
$\bar p_t^{(j)} = \mathrm{clip}(p_t^{(j)},\epsilon,1-\epsilon)$ with
a small constant $\epsilon$, and define the local LLR as
\begin{equation}
  \Lambda_t^{(j)} =
  \log\frac{\bar p_t^{(j)}}{1-\bar p_t^{(j)}}.
  \label{eq:local_llr}
\end{equation}
A positive value supports bit value $1$, a negative value supports
bit value $0$, and the magnitude measures local evidence strength.
The clipping bound directly controls the bounded-increment property
used in our robustness analysis (Section~\ref{subsec:robustness_framework}).
A single local LLR can be weak under editing or source-verifier
mismatch, so PURA uses it as soft evidence rather than as a
standalone decision.

\noindent\textbf{LLR Aggregation and Per-Bit Decoding.}
PURA aggregates all local LLRs assigned to the same encoded payload
position $i \in \{0, \dots, |c|{-}1\}$ according to the cyclic
indexing rule:
\begin{equation}
  \Lambda_{\mathrm{total}}^{(i)} =
  \sum_{\substack{t, j \\
  (\mathrm{idx}(t)+j) \bmod |c| = i}}
  \Lambda_t^{(j)}.
  \label{eq:global_llr}
\end{equation}
This aggregate LLR is the main recovery statistic: its sign gives the
bit estimate and its magnitude gives the confidence. We decode each
bit as
\begin{equation}
  \hat{c}_i =
  \mathbb{I}\!\left(
    \Lambda_{\mathrm{total}}^{(i)} \ge 0
  \right).
  \label{eq:hard_decision}
\end{equation}
Weak but aligned local evidence is reinforced by the summation, while
unstructured noise tends to cancel out.

\noindent\textbf{Reliability-Guided Repair.}
The magnitude $\bigl|\Lambda_{\mathrm{total}}^{(i)}\bigr|$ also
provides a bit-level confidence score. Since each block in $c$
contains one parity bit, a failed parity check identifies a block
with inconsistent recovery. For the $q$-th block
($q\in\{1,\dots,B_{\mathrm{blk}}\}$), let
$\mathcal{B}_q\subseteq\{0,\dots,|c|-1\}$ be its bit positions in
$c$. PURA repairs such a block by flipping its least reliable bit:
\begin{equation}
  i^* = \operatorname*{arg\,min}_{i \in \mathcal{B}_q}
        \bigl|\Lambda_{\mathrm{total}}^{(i)}\bigr|,
  \qquad \hat{c}_{i^*} \leftarrow 1 - \hat{c}_{i^*}.
  \label{eq:spc_repair}
\end{equation}
This rule is matched to the upstream design: Gray coding makes local
sector drift more likely to cause a single-bit error, and the
aggregated LLR magnitude identifies the least reliable bit in the
block. After repair, the verifier removes the parity bits and
returns the recovered payload.

Overall, PURA forms a closed embedding and extraction loop.
Entropy-aware scheduling selects informative positions. Gray-coded
anchoring turns payload bits into local latent shifts without
changing token probabilities. Interval-based extraction keeps
uncertainty instead of discarding it through hard decisions. LLR
aggregation converts many weak local signals into reliable payload
estimates, and lightweight parity repair corrects the most likely
residual errors.

\section{Theoretical Analysis}
\label{sec:theoretical_analysis}

\subsection{Computational Unbiasedness}
\label{subsec:comp_unbiasedness}

Recall Equations~\eqref{eq:anchor} and~\eqref{eq:secret_sampling}. PURA generates each next token by inverse transform sampling over a secretly permuted vocabulary order, with anchor $A_t = U_t \oplus \Delta_t$, where $U_t \in [0,1)$ is a keyed pseudo-random phase derived from the local context and $\Delta_t$ is a payload-dependent offset. In skip mode, namely when $a_t = 0$, PURA falls back to baseline sampling $x_t \sim \pi_t(\cdot \mid x_{<t})$.

\begin{theorem}[Computational unbiasedness for a single invocation]
\label{thm:comp_unbiasedness}
Let $\mathsf{Gen}_{\mathrm{PURA}}$ and $\mathsf{Gen}_{\mathrm{Baseline}}$ denote the watermarked generator and the unwatermarked generator, respectively. Fix an initial prompt $x_{\le 0}$ and a payload. Assume that all keyed derivations use domain-separated secure pseudorandom functions and that the secret key is sampled uniformly. Then, for any probabilistic polynomial-time distinguisher $\mathcal{D}$,
\[
\Bigl|\Pr[\mathcal{D}(\mathsf{Gen}_{\mathrm{PURA}})=1]
-\Pr[\mathcal{D}(\mathsf{Gen}_{\mathrm{Baseline}})=1]\Bigr|
\le \mathrm{negl}(\lambda).
\]
\end{theorem}

\begin{proof}[Proof sketch]
Replacing the domain-separated keyed derivations by random functions changes any efficient distinguisher's view only negligibly. By the repeated-context skipping rule in Equation~\eqref{eq:uniqueness}, each watermarked step then uses a fresh seed-hash input, so $U_t$ is uniform on $[0,1)$. By rotational invariance from Equation~\eqref{eq:rotation}, $U_t \oplus \Delta_t$ remains uniform, and inverse transform sampling exactly reproduces the original next-token distribution. Full details appear in Appendix~\ref{app:onecall_unbiased}.
\end{proof}

\subsection{Robustness Framework and Analysis}
\label{subsec:robustness_framework}

We analyze recovery for one fixed encoded payload position $i$ on the
verifier-observed sequence $\mathbf{x}'=(x'_1,\dots,x'_{L'})$, which
may equal the original text $\mathbf{x}$ when no editing occurs. Let
\[
  \mathcal{T}_i =
  \{\,t\in\{1,\dots,L'\}: \exists j\in\{0,\dots,a_t-1\},
  (\mathrm{idx}(t)+j)\bmod |c|=i\,\}
\]
be the set of verifier-side token steps that contribute evidence to
bit $i$. Since $a_t < |c|$ in all our settings, for each
$t\in\mathcal{T}_i$ this local coordinate is unique. Let
$T=|\mathcal{T}_i|$, and write $\Lambda_t$ for the corresponding
local LLR contributed by step $t$ to encoded payload position $i$.
We drop the index $i$ in the rest of this subsection and state the
bounds for a fixed evidence count $T$.

\noindent\textbf{Probabilistic preliminaries.}
We first recall the standard notions used in the analysis. A filtration $\{\mathcal{G}_t\}_{t\ge 0}$ is an increasing sequence of $\sigma$-algebras that represents the information revealed up to time $t$. A stochastic process $\{X_t\}$ is a martingale with respect to $\{\mathcal{G}_t\}$ if $X_t$ is $\mathcal{G}_t$-measurable, $\mathbb{E}[|X_t|]<\infty$, and
\[
  \mathbb{E}[X_t \mid \mathcal{G}_{t-1}] = X_{t-1};
\]
equivalently, the next increment has zero conditional mean given the past. Azuma--Hoeffding inequality states that a martingale with bounded increments concentrates around its initial value~\cite{roch_mdp_2024}. In the interval form, if the increment $X_t-X_{t-1}$ lies almost surely in an interval of length $c_t$, then for any $\delta>0$,
\begin{equation}
\label{eq:interval_azuma}
  \Pr[X_T-X_0 \le -\delta]
  \le
  \exp\!\left(
    -\frac{2\delta^2}{\sum_{t=1}^{T} c_t^2}
  \right).
\end{equation}
We will construct such a martingale from centered watermark evidence.

\noindent\textbf{Bounded local evidence via clipping.}
Recall from Equation~\eqref{eq:local_llr} that the local LLR
$\Lambda_t^{(j)}$ is obtained by clipping the bit-$1$ posterior to
$[\epsilon,1-\epsilon]$ before taking the log ratio. This clipping
yields the uniform bound
\begin{equation}
\label{eq:clipped_llr_bound}
  \bigl|\Lambda_t^{(j)}\bigr|
  \le
  B_\epsilon
  \triangleq
  \log\frac{1-\epsilon}{\epsilon},
\end{equation}
which provides the bounded-increment condition required by
Azuma--Hoeffding. The clipping is applied only during extraction
and does not affect the embedding distribution.

Let $c_i\in\{0,1\}$ denote the true encoded bit at position $i$. We
align the sign of each local LLR with $c_i$ and normalize it as
\begin{equation}
\label{eq:aligned_z}
  Z_t
  \triangleq
  \frac{(2c_i-1)\Lambda_t}{B_\epsilon}
  \in [-1,1],
\end{equation}
so that $Z_t>0$ supports the correct bit while $Z_t<0$ is misleading. The normalized aggregate statistic is
\begin{equation}
\label{eq:stat_LT}
  \mathcal{L}
  \triangleq
  \sum_{t\in\mathcal{T}_i} Z_t .
\end{equation}
Under the decision rule in Equation~\eqref{eq:hard_decision}, the
bit-error event is contained in $\{\mathcal{L} \le 0\}$.

\noindent\textbf{Martingale construction.}
Let $\{\mathcal{G}_t\}_{t\ge 0}$ be the verifier-side filtration over
token steps. For each $t\in\mathcal{T}_i$, define
\begin{equation}
\label{eq:mt_def}
  m_t \triangleq \mathbb{E}[Z_t \mid \mathcal{G}_{t-1}],
\end{equation}
and let
\begin{equation}
\label{eq:xt_def}
  X_T \triangleq \sum_{t\in\mathcal{T}_i} (Z_t - m_t).
\end{equation}
Indexing $\mathcal{T}_i$ in chronological order, the partial sums of
$\{Z_t - m_t\}_{t\in\mathcal{T}_i}$ form a martingale: by definition
of $m_t$,
\[
  \mathbb{E}[Z_t - m_t \mid \mathcal{G}_{t-1}]
  = \mathbb{E}[Z_t \mid \mathcal{G}_{t-1}] - m_t
  = m_t - m_t
  = 0,
\]
and $Z_t - m_t \in [-1 - m_t,\, 1 - m_t]$ has length $2$ since
$Z_t \in [-1, 1]$. Thus Equation~\eqref{eq:interval_azuma} applies
with $c_t = 2$ over the $T = |\mathcal{T}_i|$ evidence steps, and the
aggregate decomposes as
\begin{equation}
\label{eq:LT_xt_relation}
  \mathcal{L} = \sum_{t\in\mathcal{T}_i} m_t + X_T ,
\end{equation}
so it remains to identify the predictable drift $\sum_t m_t$.

\noindent\textbf{Structural survival.}
For each generation step $t\in\mathcal{T}_i$, let
$\mathsf{Val}_t \in \{0,1\}$ indicate whether the evidence at step
$t$ is synchronized, meaning that both the length-$W$ context
window and the current token survive editing. Define the
bit-level survival rate
\begin{equation}
\label{eq:rho_bit_def}
  \rho
  \triangleq
  \frac{1}{T}
  \sum_{t\in\mathcal{T}_i}
  \mathbb{E}[\mathsf{Val}_t\mid\mathcal{G}_{t-1}],
\end{equation}
which captures the average conditional rate at which synchronized
evidence remains available to bit $i$ after editing. The reliability
bounds below are stated directly in terms of $\rho$.

\noindent\textbf{Main bound under TM0.}
Define the synchronized conditional mean
\begin{equation}
\label{eq:mu_defs}
  \mu_c
  \triangleq
  \frac{
  \sum_{t\in\mathcal{T}_i}
  \mathbb{E}[\mathsf{Val}_t Z_t\mid\mathcal{G}_{t-1}]
  }{
  \sum_{t\in\mathcal{T}_i}
  \mathbb{E}[\mathsf{Val}_t\mid\mathcal{G}_{t-1}]
  } .
\end{equation}
The quantity $\mu_c\in(0,1]$ is the average aligned evidence strength on synchronized units. Under TM0, unsynchronized observations do not systematically favor either bit value, so their LLR contributions are centered. Appendix~\ref{app:tm0_centeredness} formalizes this symmetry argument. Hence
\begin{equation}
\label{eq:tm0_centeredness}
\mathbb{E}[Z_t \mid \mathcal{G}_{t-1},\mathsf{Val}_t=0]
  =
  0 ,
\end{equation}
If no synchronized unit remains, we take $\rho\mu_c=0$.

By conditional decomposition,
\[
m_t
=
\mathbb{E}[\mathsf{Val}_t Z_t\mid\mathcal{G}_{t-1}]
+
\mathbb{E}[(1-\mathsf{Val}_t)Z_t\mid\mathcal{G}_{t-1}].
\]
Equation~\eqref{eq:tm0_centeredness} makes the second term zero.
Therefore, the predictable drift satisfies
\begin{equation}
\label{eq:tm0_drift_sum}
  \sum_{t\in\mathcal{T}_i} m_t
  =
  T\rho\mu_c .
\end{equation}
For the probability bound, let $\Gamma>0$ be a uniform lower bound on
the average drift:
\begin{equation}
\label{eq:gamma_def}
  \rho\mu_c \ge \Gamma > 0
  \quad\text{almost surely}.
\end{equation}

\begin{proposition}[Exponential reliability under TM0]
\label{prop:azuma_tool_tm0}
If Equation~\eqref{eq:gamma_def} holds, then
\begin{equation}
\label{eq:tm0_bound}
  \Pr(\mathcal{L} \le 0)
  \le
  \exp\!\left(
  -\frac{T\Gamma^2}{2}
  \right).
\end{equation}
\end{proposition}

\begin{proof}
If $\mathcal{L}\le 0$, then by Equations~\eqref{eq:LT_xt_relation} and~\eqref{eq:tm0_drift_sum},
\[
  X_{T}
  =
  \mathcal{L}
  -
  \sum_{t\in\mathcal{T}_i}m_t
  \le
  -T\rho\mu_c
  \le
  -T\Gamma .
\]
Applying Equation~\eqref{eq:interval_azuma} with $\delta=T\Gamma$ and $c_t=2$ yields
\[
  \Pr(\mathcal{L}\le 0)
  \le
  \Pr(X_{T}\le -T\Gamma)
  \le
  \exp\!\left(
  -\frac{2T^2\Gamma^2}{4T}
  \right)
  =
  \exp\!\left(
  -\frac{T\Gamma^2}{2}
  \right).
\]
\end{proof}

\noindent\textbf{Extension to adaptive attackers under TM1.}
Under TM1, the attacker knows the watermarking algorithm but not the key and receives no detection feedback. In this setting, unsynchronized evidence need not be centered. We summarize its average signed contribution by
\begin{equation}
\label{eq:tm1_centeredness}
  \mu_{\mathrm{adv}}
  \triangleq
  \frac{
  \sum_{t\in\mathcal{T}_i}
  \mathbb{E}[(1-\mathsf{Val}_t)Z_t\mid\mathcal{G}_{t-1}]
  }{
  \sum_{t\in\mathcal{T}_i}
  \mathbb{E}[1-\mathsf{Val}_t\mid\mathcal{G}_{t-1}]
  } .
\end{equation}
We do not assume a fixed sign for $\mu_{\mathrm{adv}}$: a negative value means that the attack creates systematically misleading evidence, while a value near zero means that the attack mainly acts through structural desynchronization. Thus, $\mu_{\mathrm{adv}}$ captures whether desynchronized observations merely add noise or also introduce a directional bias against the true bit.
If all units are synchronized, we take
$(1-\rho)\mu_{\mathrm{adv}}=0$.

For TM1, let $\Gamma>0$ satisfy
\begin{equation}
\label{eq:tm1_gamma_def}
  \rho\mu_c
  +
  (1-\rho)\mu_{\mathrm{adv}}
  \ge \Gamma > 0
  \quad\text{almost surely},
\end{equation}
and the predictable drift satisfies
\begin{equation}
\label{eq:tm1_drift_sum}
  \sum_{t\in\mathcal{T}_i} m_t
  =
  T\rho\mu_c
  +
  T(1-\rho)\mu_{\mathrm{adv}}
  \ge
  T\Gamma .
\end{equation}

\begin{proposition}[Drift-penalized bound under TM1]
\label{prop:tm1_bound}
If Equation~\eqref{eq:tm1_gamma_def} holds, then
\begin{equation}
\label{eq:tm1_bound}
  \Pr(\mathcal{L} \le 0)
  \le
  \exp\!\left(
  -\frac{T\Gamma^2}{2}
  \right).
\end{equation}
\end{proposition}

\begin{proof}
If $\mathcal{L}\le 0$, then by Equations~\eqref{eq:LT_xt_relation} and~\eqref{eq:tm1_drift_sum}, $X_{T}\le -T\Gamma$. Applying Equation~\eqref{eq:interval_azuma} with $\delta=T\Gamma$ and $c_t=2$ gives the bound.
\end{proof}

The TM0 and TM1 results share the same form and are controlled by a positive lower bound $\Gamma$ on the average drift. Under TM0, unsynchronized evidence is centered, so the drift is $\rho\mu_c$. Under TM1, the attacker can reduce it either by lowering the structural survival rate $\rho$ or by pushing $\mu_{\mathrm{adv}}$ below zero through directionally misleading evidence. The bound therefore separates robustness into a structural channel and a directional evidence channel, and explains why soft aggregation remains effective when enough synchronized evidence survives and the remaining unsynchronized evidence does not impose a strong negative drift.
\section{Experiments}

We evaluate PURA from three perspectives: high-capacity message recovery, robustness to post-editing and semantic rewriting attacks, and verification efficiency in asymmetric deployment. We also conduct component-level ablations to isolate the effect of key design choices. Additional results appear in Appendix~\ref{app:ex}.

\subsection{Experimental Setup}

\noindent\textbf{Datasets.}
We evaluate PURA on C4~\cite{rael_exploring_nodate}, Essays~\cite{ivypanda2024}, and OpenGen~\cite{merity_pointer_2016}. Unless stated otherwise, C4 is the default dataset.

\noindent\textbf{Models and source-proxy pairs.}
To evaluate asymmetric verification, we use source-proxy model pairs
from four model families, listed as
source\,/\,proxy: Llama-3
(8B\,/\,3B)~\cite{grattafiori_llama_2024}, Gemma-2
(9B\,/\,2B)~\cite{gemma2_2024}, BLOOM
(7.1B\,/\,3B)~\cite{bloom_2022}, and Pythia
(6.9B\,/\,2.8B)~\cite{pythia_icml_2023}. In the asymmetric setting,
watermarked text is generated by the source model and verified by the
proxy model. In the symmetric setting, the source model itself acts
as the verifier. All pairs share the same tokenizer and vocabulary.
Unless noted otherwise, we report asymmetric results and use the
Llama-3 family by default. 

\noindent\textbf{Baselines.}
We compare PURA with recent multi-bit watermarking methods and divide them into two groups according to whether they bias the sampling distribution. The biased baselines are RS-BH~\cite{qu2024via_code}, MPAC~\cite{yoo2023MPAC}, and CTWL~\cite{DBLP:conf/iclr/WangYC0LM0024}. The unbiased baselines are CycleShift~\cite{fernandez2023cycleshift}, BiMark~\cite{feng2025bimark}, and StealthInk~\cite{jiang2025stealthink}.

\noindent\textbf{Configuration.}
We standardize the context window size to $W=2$ for methods that use context-aware hashing. Unless stated otherwise, we embed 12--36 payload bits, excluding parity bits, per 200 generated tokens. PURA uses the cosine kernel by default, with 3 parity blocks for 12--30 bits and 4 parity blocks for 36 bits. The sampling temperature is $1.0$, and the maximum generation length is 200 tokens. We report match rate, bit accuracy, and per-sample verification time. Match rate is the fraction of test samples whose recovered payload exactly matches the embedded payload. Bit accuracy is the fraction of correctly recovered payload bits, computed for each sample and then averaged over the test set. Baseline hyperparameters follow their original papers or implementations.

\subsection{Main Results: Capacity, Accuracy, and Efficiency}
\label{subsec:main_results}

\begin{table*}[t]
    \centering
    \caption{Capacity, recovery accuracy, and per-sample verification time across payload lengths. 
    Best results are shown in bold and second-best results are underlined. 
    ``/'' marks an infeasible setting and is excluded from ranking.}
    \label{tab:comparison_table}
    \resizebox{\linewidth}{!}{
    \begin{tabular}{l|ccc|ccc|ccc|ccc|ccc}
        \toprule
        \multirow{2}{*}{Method}
        & \multicolumn{3}{c|}{12 bits}
        & \multicolumn{3}{c|}{18 bits}
        & \multicolumn{3}{c|}{24 bits}
        & \multicolumn{3}{c|}{30 bits}
        & \multicolumn{3}{c}{36 bits} \\
        \cmidrule(lr){2-4}\cmidrule(lr){5-7}
        \cmidrule(lr){8-10}\cmidrule(lr){11-13}\cmidrule(lr){14-16}
        & \shortstack{Match\\(\%)} & \shortstack{Bit Acc.\\(\%)} & \shortstack{Time\\(s)}
        & \shortstack{Match\\(\%)} & \shortstack{Bit Acc.\\(\%)} & \shortstack{Time\\(s)}
        & \shortstack{Match\\(\%)} & \shortstack{Bit Acc.\\(\%)} & \shortstack{Time\\(s)}
        & \shortstack{Match\\(\%)} & \shortstack{Bit Acc.\\(\%)} & \shortstack{Time\\(s)}
        & \shortstack{Match\\(\%)} & \shortstack{Bit Acc.\\(\%)} & \shortstack{Time\\(s)} \\
        \midrule
        \multicolumn{16}{l}{Biased baselines} \\
        RS-BH~\cite{qu2024via_code} ($\delta=3.0$)
        & 98.7 & 99.64 & \second{0.07} & 97.9 & 99.46 & 0.09
        & 95.9 & 98.83 & 0.09 & 90.5 & 97.48 & 0.11
        & 81.4 & 95.82 & 0.18 \\
        RS-BH~\cite{qu2024via_code} ($\delta=2.0$)
        & 93.3 & 98.26 & \second{0.07} & 82.7 & 95.17 & \second{0.07}
        & 60.4 & 88.82 & 0.09 & 35.4 & 81.60 & 0.11
        & 20.8 & 80.45 & 0.17 \\
        MPAC~\cite{yoo2023MPAC} ($\delta=3.0$)
        & 99.1 & 99.77 & 0.08 & 96.2 & 99.52 & 0.08
        & 87.2 & 99.02 & \second{0.08} & 70.6 & 98.17 & \second{0.08}
        & 52.3 & 96.05 & \second{0.09} \\
        MPAC~\cite{yoo2023MPAC} ($\delta=2.0$)
        & 91.8 & 98.77 & 0.08 & 71.0 & 96.99 & 0.08
        & 41.3 & 94.71 & \second{0.08} & 18.5 & 92.09 & \second{0.08}
        & 5.6  & 85.30 & \second{0.09} \\
        CTWL~\cite{DBLP:conf/iclr/WangYC0LM0024} ($\delta=2.0$)
        & \second{99.3} & 99.67 & 6.35 & 97.0 & 98.57 & 7.12
        & / & / & / & / & / & / & / & / & / \\
        \midrule
        \multicolumn{16}{l}{Unbiased baselines} \\
        CycleShift~\cite{fernandez2023cycleshift}
        & 98.6 & 99.28 & 2.97 & 97.8 & 98.89 & 15.02
        & / & / & / & / & / & / & / & / & / \\
        BiMark~\cite{feng2025bimark}
        & 97.2 & 99.45 & \second{0.07} & 87.8 & 98.94 & \second{0.07}
        & 65.0 & 97.62 & 0.09 & 50.0 & 96.97 & 0.12
        & 28.9 & 94.12 & 0.13 \\
        StealthInk~\cite{jiang2025stealthink}
        & 60.0 & 96.21 & 0.08 & 25.5 & 93.29 & 0.08
        & 6.5  & 89.36 & \second{0.08} & 2.0  & 88.20 & \second{0.08}
        & 0.5  & 80.15 & \second{0.09} \\
        \midrule
        \multicolumn{16}{l}{Our method, unbiased} \\
        PURA (asymmetric)
        & 99.1 & \second{99.93} & \best{0.05}
        & \second{98.6} & \second{99.71} & \best{0.05}
        & \second{97.1} & \second{99.51} & \best{0.05}
        & \second{95.8} & \second{99.28} & \best{0.05}
        & \second{91.7} & \second{99.02} & \best{0.05} \\
        PURA (symmetric)
        & \best{99.8} & \best{99.97} & 0.12
        & \best{99.6} & \best{99.95} & 0.12
        & \best{99.4} & \best{99.84} & 0.12
        & \best{99.0} & \best{99.82} & 0.12
        & \best{97.2} & \best{99.39} & 0.12 \\
        \bottomrule
    \end{tabular}
    }
\end{table*}

We evaluate PURA across different payload lengths, datasets, and model families. The results are summarized in Table~\ref{tab:comparison_table} and Figure~\ref{fig:generalizability}.

\begin{figure*}[!tbp]
\centering
\includegraphics[width=\linewidth]{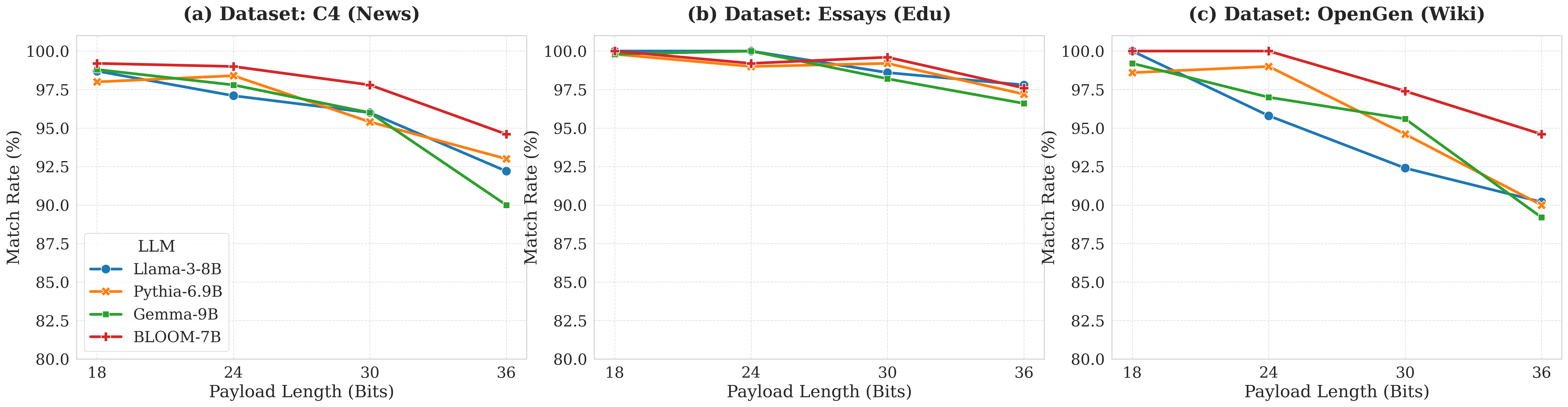}
\caption{Extraction accuracy as a function of payload length across datasets and model families.}
\Description{Three line charts for C4, Essays, and OpenGen show match rates across four language models. Accuracy declines as payload length increases from 18 to 36 bits, with Essays remaining the most stable and BLOOM generally strongest on C4 and OpenGen.}
\label{fig:generalizability}
\end{figure*}

\noindent\textbf{Scaling to high payloads.}
We sweep the payload length from 12 to 36 bits per 200 tokens. Several baselines remain effective at low payloads, but recovery degrades sharply as the payload grows. At 36 bits, StealthInk reaches only 0.5\% match rate, BiMark drops to 28.9\%, and the strongest biased baseline, RS-BH with $\delta=3.0$, falls to 81.4\%. PURA-asymmetric retains 91.7\% match rate and 99.02\% bit accuracy in the same regime, showing that soft-evidence aggregation remains stable under high payload density.

\noindent\textbf{Generalization across models and domains.}
Figure~\ref{fig:generalizability} shows that PURA maintains a high match rate across all configurations in four model families and three datasets. Within each family, the alignment between proxy and source models is strong enough to support stable extraction across parameter scales. Essays is the most stable domain, with match rate above 95\% even at 36 bits. This suggests that the entropy gate retains more usable evidence steps on this domain, giving the soft-evidence aggregator more aligned observations for recovery.

\noindent\textbf{Verification efficiency.}
PURA achieves millisecond-level verification, and the runtime is nearly invariant to payload length, which is comparable to black-box detection baselines. By contrast, the decoding cost of CycleShift and CTWL grows sharply with payload size. For CTWL, per-sample decoding reaches the minute level beyond 18 bits, so we do not evaluate these methods at higher payloads. The efficiency of PURA comes from its vectorized detector: one GPU forward pass produces all next-token distributions, after which lightweight interval integration and simple reductions recover the full payload.

\subsection{Text Quality and Stealth}
\label{subsec:quality_stealthiness}

A core requirement of unbiased watermarking is that it remain imperceptible to human readers and statistically indistinguishable to automated detectors. We evaluate the distributional similarity and stealth of PURA relative to both biased and unbiased methods. The results are shown in Table~\ref{tab:downstream_stealthiness}.

For detectability, we use a RoBERTa-base~\cite{yang2020tscsw} binary classifier to distinguish watermarked text. PURA's classifier accuracy is close to random guessing, at 50.23\% versus 50.00\% for unwatermarked text, whereas biased schemes are more separable, e.g., MPAC with $\delta=3$ reaches 68.28\%. On summarization and English-to-Romanian translation~\cite{ICLR2024_c5b00c5b}, PURA remains within 0.16 BERTScore, 0.09 ROUGE-1, and 0.43 BLEU of the unwatermarked baseline. In contrast, RS-BH with $\delta=3$ loses 3.22 BERTScore and 5.10 BLEU. BiMark and StealthInk achieve similarly strong quality and stealth, but their high-payload recovery in Table~\ref{tab:comparison_table} is substantially weaker than PURA.

\begin{table}[t]
    \centering
    \caption{Stealth and downstream task quality under different watermarking methods.}
    \label{tab:downstream_stealthiness}
    \resizebox{\linewidth}{!}{
    \begin{tabular}{l|c|ccc|cc}
        \toprule
        \multirow{2}{*}{Method}
        & Stealth
        & \multicolumn{3}{c|}{Text Summarization}
        & \multicolumn{2}{c}{Machine Translation} \\
        \cmidrule(lr){2-2}\cmidrule(lr){3-5}\cmidrule(lr){6-7}
        & Accuracy (\%) $\downarrow$
        & PPL $\downarrow$ & BERT $\uparrow$ & R-1 $\uparrow$
        & BERT $\uparrow$ & BLEU $\uparrow$ \\
        \midrule
        Unwatermarked
        & 50.00 & 5.028 & 32.72 & 38.57 & 56.18 & 22.12 \\
        \midrule
        RS-BH ($\delta=2.0$)
        & 56.60 & 6.229 & 31.12 & 37.14 & 53.97 & 19.83 \\
        RS-BH ($\delta=3.0$)
        & 62.30 & 8.250 & 29.50 & 35.80 & 51.08 & 17.02 \\
        MPAC ($\delta=2.0$)
        & 56.21 & 5.489 & 32.24 & 38.14 & 53.95 & 19.75 \\
        MPAC ($\delta=3.0$)
        & 68.28 & 7.330 & 31.11 & 36.90 & 52.80 & 19.08 \\
        \midrule
        BiMark
        & 49.75 & 5.067 & 32.73 & 38.54 & 56.53 & 21.96 \\
        StealthInk
        & 50.90 & 5.013 & 32.71 & 38.59 & 56.20 & 22.13 \\
        \midrule
        PURA
        & 50.23 & 5.010 & 32.56 & 38.48 & 56.02 & 21.69 \\
        \bottomrule
    \end{tabular}
    }
\end{table}

\subsection{Robustness to Post-Editing Attacks}
\label{subsec:robustness}

\begin{figure*}[t]
    \centering
    \includegraphics[width=\linewidth]{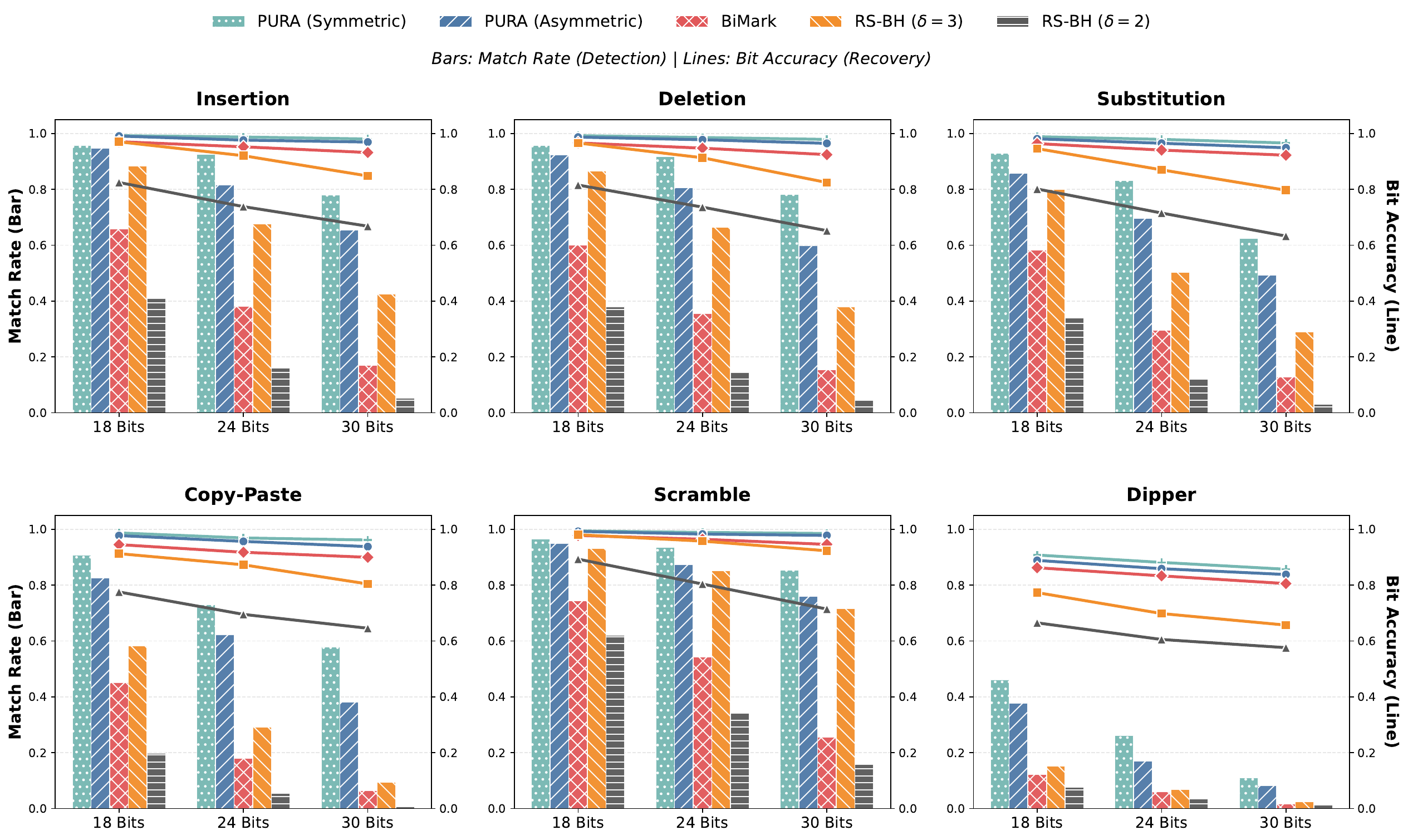}
    \caption{Match rate (bars) and bit accuracy (lines) under six editing attacks across payload lengths.}
    \Description{Six panels compare match-rate bars and bit-accuracy lines under insertion, deletion, substitution, copy-paste, scrambling, and DIPPER attacks at 18, 24, and 30 bits. Both PURA variants remain above the baselines, although match rates decrease as the payload grows.}
    \label{fig:robustness_attacks}
\end{figure*}

\begin{figure*}[t]
    \centering
    \includegraphics[width=\linewidth]{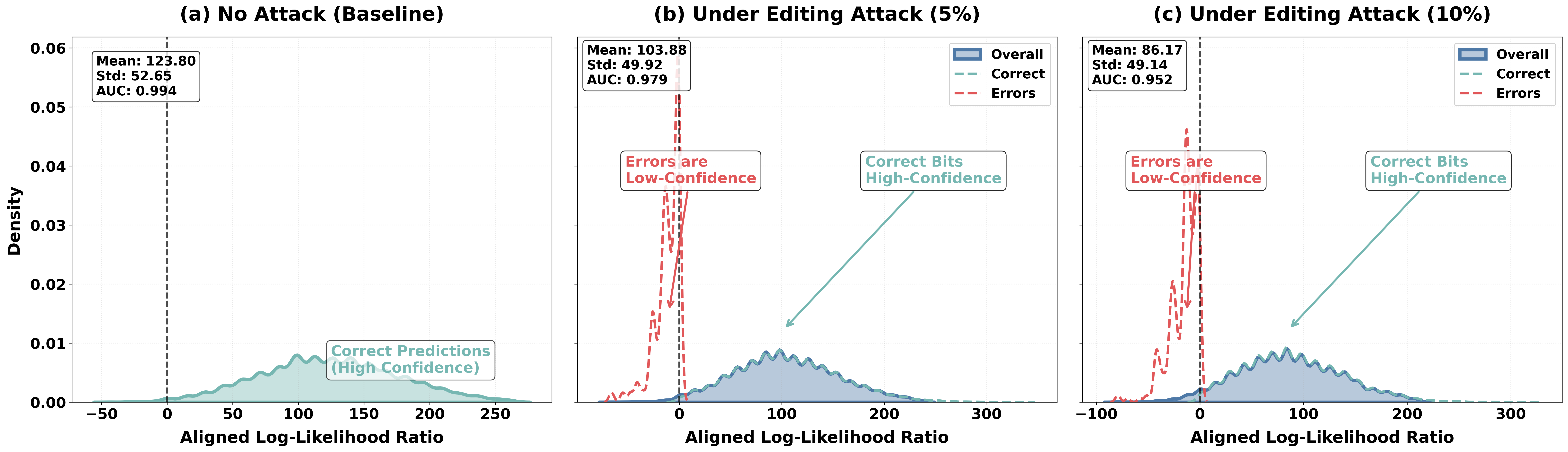}
    \caption{Aligned-LLR density for correct and incorrect bits at 24-bit payload under increasing edit strength. The area under the ROC curve (AUC) is $\Pr(|\Lambda_{\mathrm{correct}}| > |\Lambda_{\mathrm{incorrect}}|)$.}
\Description{Three density plots show aligned LLRs with no attack and with editing ratios of 5 and 10 percent. Correct bits remain concentrated at positive high-confidence values, while errors cluster near or below zero, and the separation decreases with attack strength but remains pronounced.}
    \label{fig:llr_density}
\end{figure*}

In realistic deployment, watermarked text often undergoes post-processing that ranges from local edits to global semantic rewriting. We evaluate PURA under six attacks: token-level insertion, deletion, and substitution at a 10\% rate; sentence-level copy-paste, which mixes in unwatermarked samples at a ratio of 0.3; sentence shuffling; and global paraphrasing with DIPPER(20,\,20)~\cite{krishna2023paraphrasing}.

\noindent\textbf{Resilience across attack surfaces.}
Conventional baselines degrade rapidly as the payload increases. At 30 bits, both RS-BH and BiMark have near-zero match rates under DIPPER attacks. In contrast, PURA keeps bit accuracy above 85\% for all attack types even at 36 bits. This resilience comes from the combination of soft confidence scoring and Gray-coded sector mapping. Low-magnitude LLR observations behave like soft erasures in the reliability ranking: they contribute little aligned evidence and are the first candidates for parity-guided repair, while high-confidence observations drive recovery.

\noindent\textbf{Desynchronization and paraphrasing attacks.}
Deletion and copy-paste operations break the continuity of context windows and cause the verifier to lose structural synchronization. PURA mitigates this problem through the small $W=2$ context window and context-derived indexing. Even when the global order is disturbed, the verifier can still extract valid evidence from locally aligned fragments, which is consistent with the higher match rates of PURA in Figure~\ref{fig:robustness_attacks}. The DIPPER attack further tests robustness under semantic-preserving rewriting. Under such strong rewriting, PURA in the asymmetric setting retains substantially higher bit accuracy than the compared baselines.

\subsection{Analysis of Adaptive Attacks}
\label{subsec:adaptive}

We evaluate on raw 18-bit payloads with SPC repair disabled, which
exposes the soft-evidence channel directly, while parity repair would
otherwise partially mask how $\hat{\Gamma}$ translates into bit accuracy.
The reported $\hat{\rho}$, $\hat{\mu}_c$, $\hat{\mu}_{\mathrm{adv}}$,
and $\hat{\Gamma}$ are empirical counterparts of the quantities in
Section~\ref{subsec:robustness_framework}, computed as bit-level
averages over all payload positions and valid test examples, while
Match is the exact recovery rate of the whole 18-bit message.

Under TM1 we cover four existing adaptive rewriting attacks: the
Smoothing Attack~\cite{chang-etal-2025-watermark},
RLCracker~\cite{huang2025RLCracker}, Diaa et
al.~\cite{diaa2025optimizing}, and SIRA~\cite{pmlr-v267-cheng25c},
with all paraphrase-based attacks except DIPPER using
\textsc{Qwen2.5-7B-Instruct}~\cite{qwen2} as the rewriting model.
Reynolds et al.~\cite{reynolds2025breaking} and the detector-guided
variant of Zhang et al.~\cite{zhang2026character} are excluded
because their assumptions, namely a globally fixed permutation with
context-independent keys and detector feedback, are incompatible
with TM1.

\begin{promptbox}{Prompt for Paraphrase Attack}
\label{box:paraphrase_prompt}
\footnotesize
\textbf{System.}\\
You are a careful paraphraser. Rewrite only the given text. Keep the meaning, facts, names, numbers, dates, quotes, URLs, technical terms, and intent unchanged. Do not summarize, omit, add, continue, explain, or label anything.

\vspace{0.5em}
\textbf{User.}
\begin{verbatim}
Rewrite the passage below using different wording.
Requirements:
- Preserve all important facts and meaning.
- Preserve proper nouns, numbers, dates, quoted
  claims, URLs, and technical terms exactly when
  needed.
- Do not summarize.
- Keep the ordering of ideas and factual scope
  aligned with the original passage.
- Improve fluency only if it does not change
  meaning.
- Avoid copying long phrases verbatim.
- Output only the rewritten passage.
Passage:
{text}
\end{verbatim}
\end{promptbox}

\noindent\textbf{Observed failure mode.}
Table~\ref{tab:adaptive_attack_comparison} shows that TM1 rewriting
attacks mainly reduce the structural survival rate $\hat{\rho}$,
lowering Match from $84.78\%$ under random substitution to
$8.69\%$--$22.67\%$, while $\hat{\mu}_{\mathrm{adv}}$ stays within
$[-0.009,\,+0.008]$. In our evaluation, the considered TM1 attacks
do not systematically push observed intervals toward the wrong side
of the Gray-bit partition and primarily degrade recovery through
structural desynchronization. This is consistent with the
decomposition in Section~\ref{subsec:robustness_framework}: current
adaptive rewriting attacks degrade recovery primarily through
desynchronization rather than directional drift.

\begin{table}[t]
\centering
\caption{Robustness decomposition under TM0 and TM1 attacks at 18-bit raw payload, with SPC repair disabled. Margin parameters are defined in Section~\ref{subsec:robustness_framework}.}
\label{tab:adaptive_attack_comparison}
\setlength{\tabcolsep}{5pt}
\resizebox{\linewidth}{!}{%
\begin{tabular}{lcccccc}
\toprule
Configuration
& $\hat{\rho}$
& $\hat{\mu}_c$
& $\hat{\mu}_{\mathrm{adv}}$
& $\hat{\Gamma}$
& Bit Acc. (\%)
& Match (\%) \\
\midrule
\multicolumn{7}{l}{\textbf{TM0: non-adaptive reference attacks}} \\
\quad Substitution ($r=0.1$)
& 0.784 & 0.602 & $-$0.002 & 0.472 & 98.44 & 84.78 \\
\quad DIPPER~\cite{krishna2023paraphrasing}
& 0.396 & 0.512 & $+$0.004 & 0.205 & 90.79 & 37.80 \\
\midrule
\multicolumn{7}{l}{\textbf{TM1: adaptive rewriting attacks}} \\
\quad Smoothing Attack~\cite{chang-etal-2025-watermark}
& 0.371 & 0.485 & $+$0.008 & 0.185 & 85.20 & 22.67 \\
\quad RLCracker~\cite{huang2025RLCracker}
& 0.356 & 0.501 & $-$0.006 & 0.174 & 84.20 & 21.82 \\
\quad Diaa et al.~\cite{diaa2025optimizing}
& 0.290 & 0.471 & $-$0.009 & 0.130 & 78.23 & 8.69 \\
\quad SIRA~\cite{pmlr-v267-cheng25c}
& 0.355 & 0.473 & $-$0.006 & 0.164 & 80.86 & 14.50 \\
\bottomrule
\end{tabular}%
}
\end{table}

\begin{figure}[t]
    \centering
    \includegraphics[width=\linewidth]{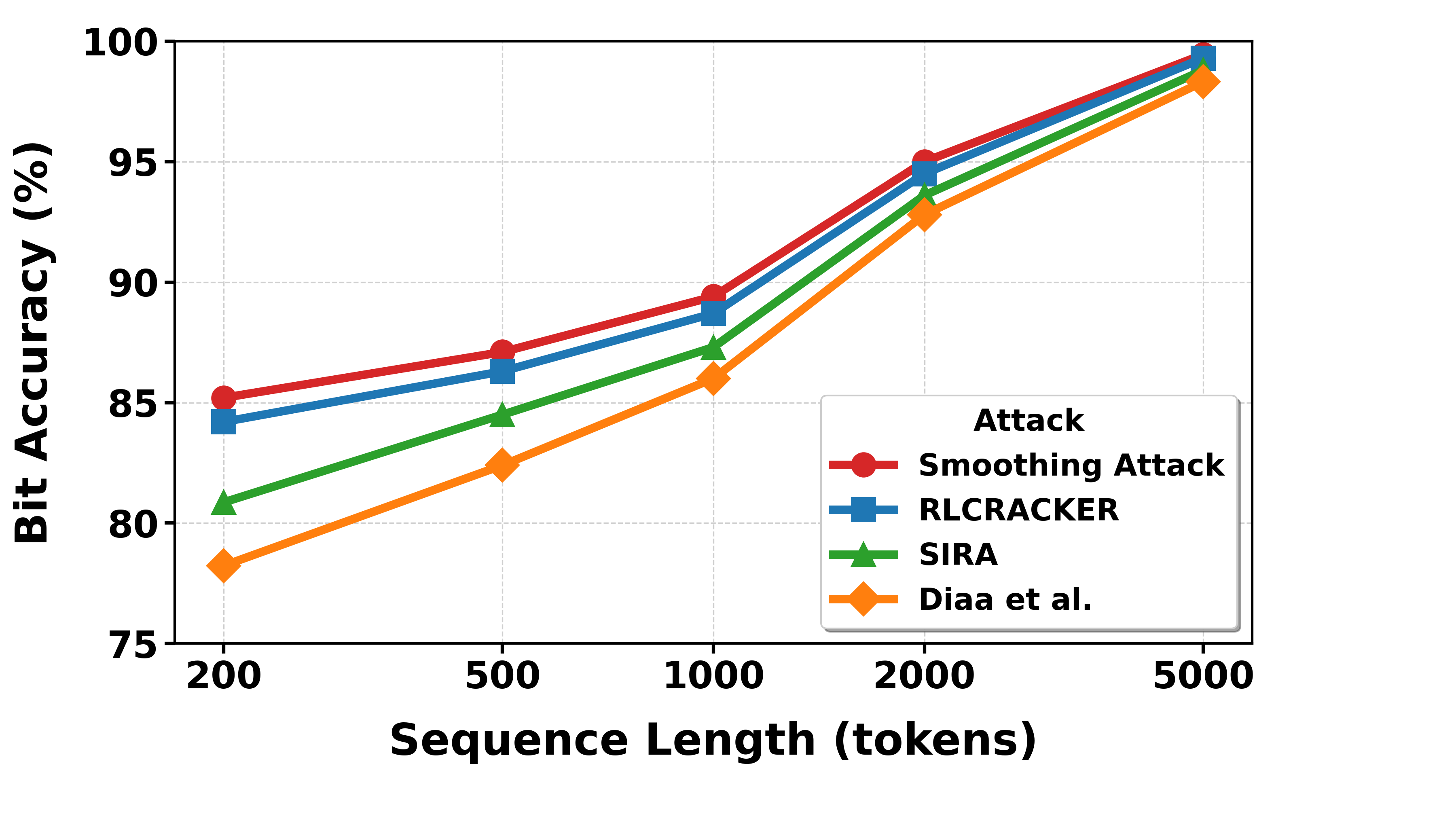}
    \caption{Bit accuracy under TM1 rewriting attacks as a function of
    sequence length.}
    \Description{Four rising curves show that bit accuracy under Smoothing Attack, RLCracker, SIRA, and Diaa et al. improves with sequence length from 200 to 5000 tokens, with all four approaching 100 percent at the longest length.}
    \label{fig:length_scaling_tm1}
\end{figure}

\noindent\textbf{Scaling with sequence length.}
We extend the four TM1 attacks to sequence lengths of 200--5000 tokens
to examine the length-dependent term in
Proposition~\ref{prop:tm1_bound}. Results are shown in
Figure~\ref{fig:length_scaling_tm1}. Bit accuracy rises monotonically
with sequence length, and all four attacks converge above $98\%$ at
5000 tokens, including the strongest one of Diaa et al., which
improves from $78.2\%$ at 200 tokens to $98.3\%$. The trend is
consistent with the exponential decay in
Proposition~\ref{prop:tm1_bound}.

\noindent\textbf{Beyond TM1.}
These experiments do not preclude stronger adaptive strategies in
principle. For example, an attacker with sufficient access and
optimization capability might try to aggregate information across many
generated samples or exploit residual key-correlated structure in order
to induce a negative drift in the directional evidence channel
$\mu_{\mathrm{adv}}$, thereby weakening the condition required by our
robustness bound. At present, however, such key-directed adaptive attacks
do not have a directly reproducible end-to-end implementation, and
turning this theoretical strategy into a practical attack remains an
open problem. To study this regime in a controlled way,
Appendix~\ref{app:oracle_attack} introduces a privileged oracle stress
test. The oracle is strictly outside TM1, but it approximates the
boundary case in which an attacker can directly search for edits that
reduce $\mu_{\mathrm{adv}}$.

\subsection{Ablation Study and Component Analysis}
\label{subsec:analysis}

\begin{figure*}[t]
    \centering
    \captionsetup[subfigure]{skip=4pt}

    \begin{subfigure}[t]{0.49\textwidth}
        \centering
        \includegraphics[width=\linewidth]{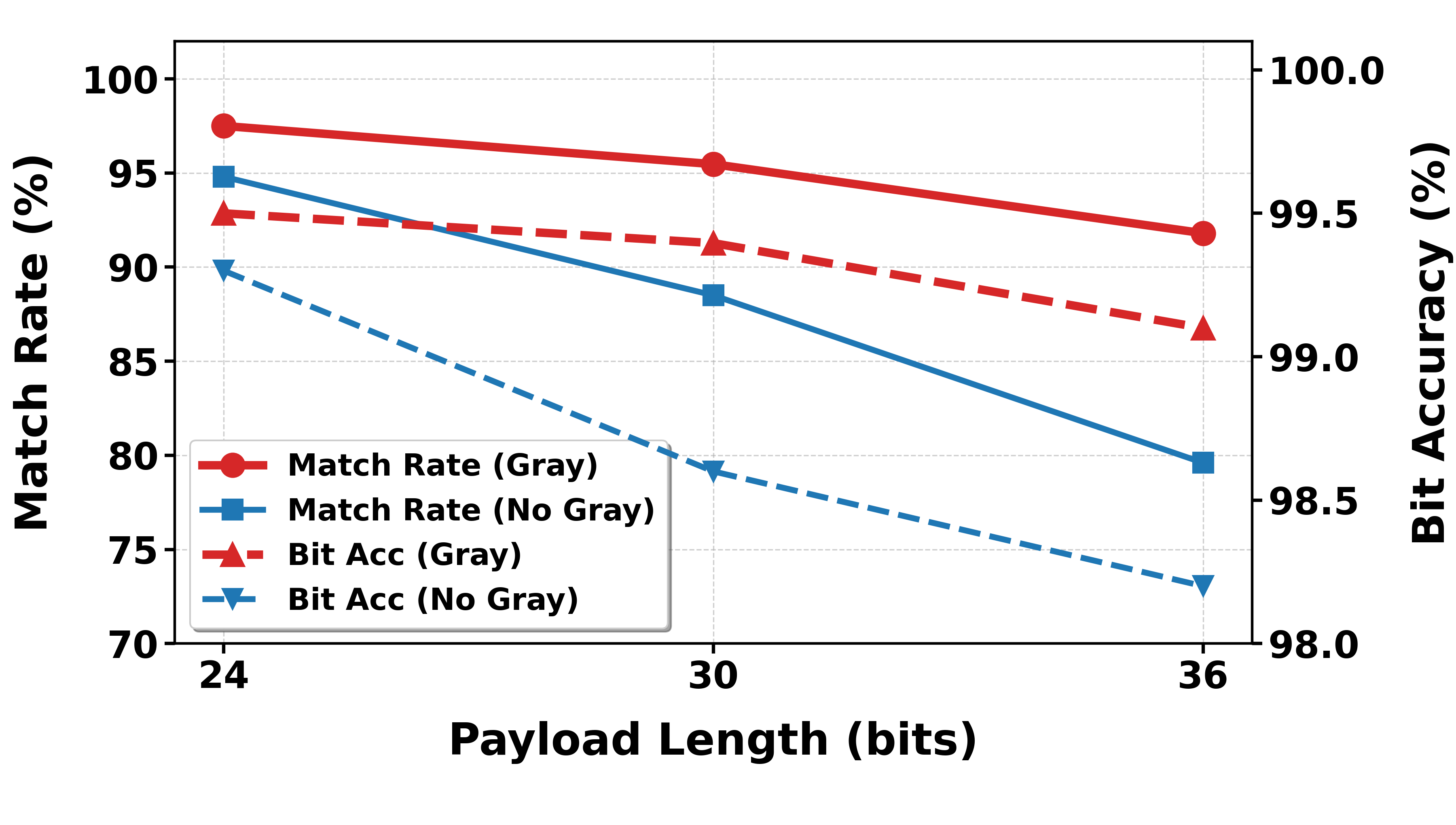}
        \caption{Gray coding vs.\ binary indexing.}
        \label{fig:gray_code_ablation}
    \end{subfigure}
    \hfill
    \begin{subfigure}[t]{0.49\textwidth}
        \centering
        \includegraphics[width=\linewidth]{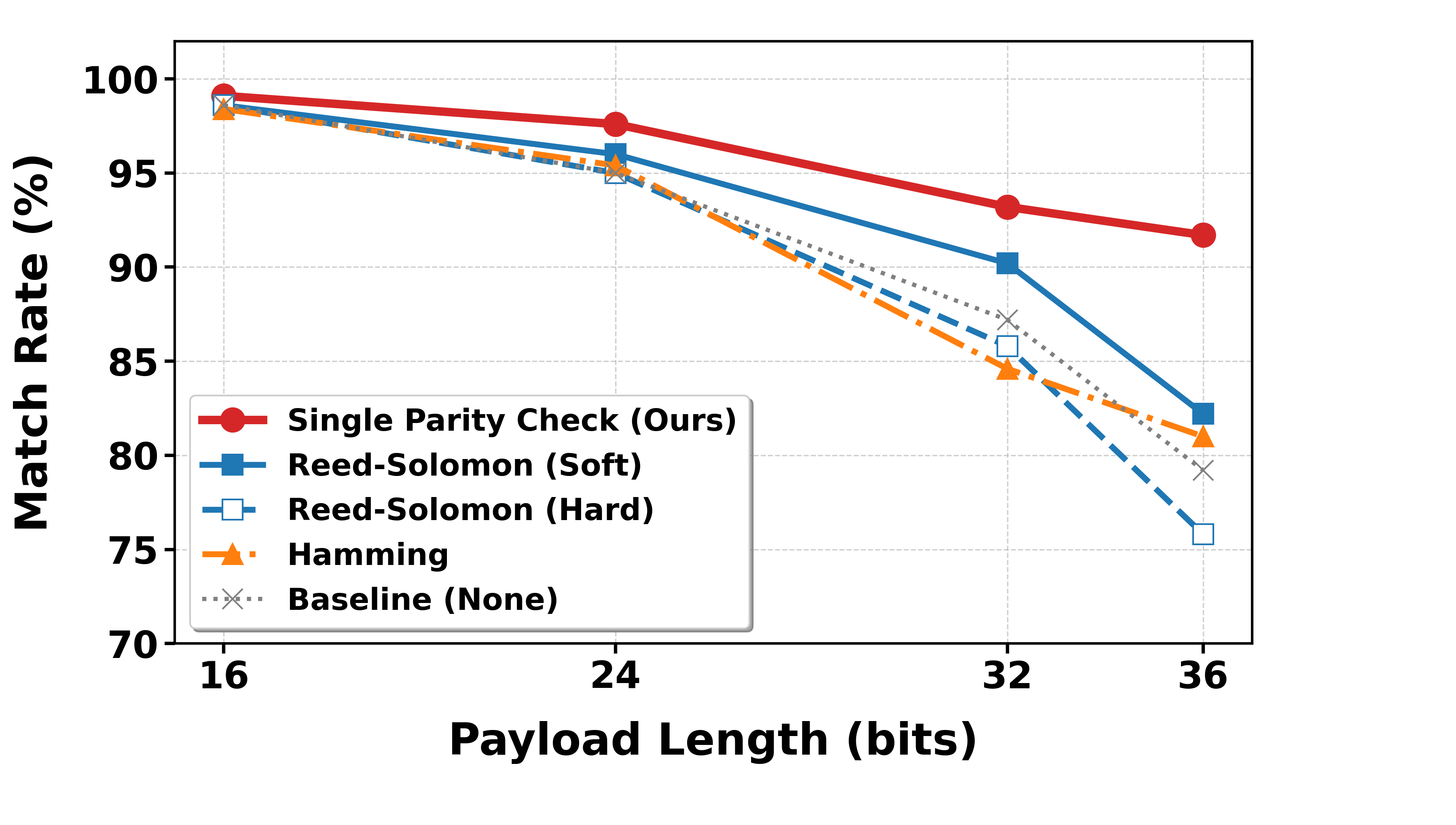}
        \caption{Error-correcting code (ECC) strategy: SPC vs.\ Hamming vs.\ RS.}
        \label{fig:ecc_ablation}
    \end{subfigure}

    \vspace{0.8em}

    \begin{subfigure}[t]{0.49\textwidth}
        \centering
        \includegraphics[width=\linewidth]{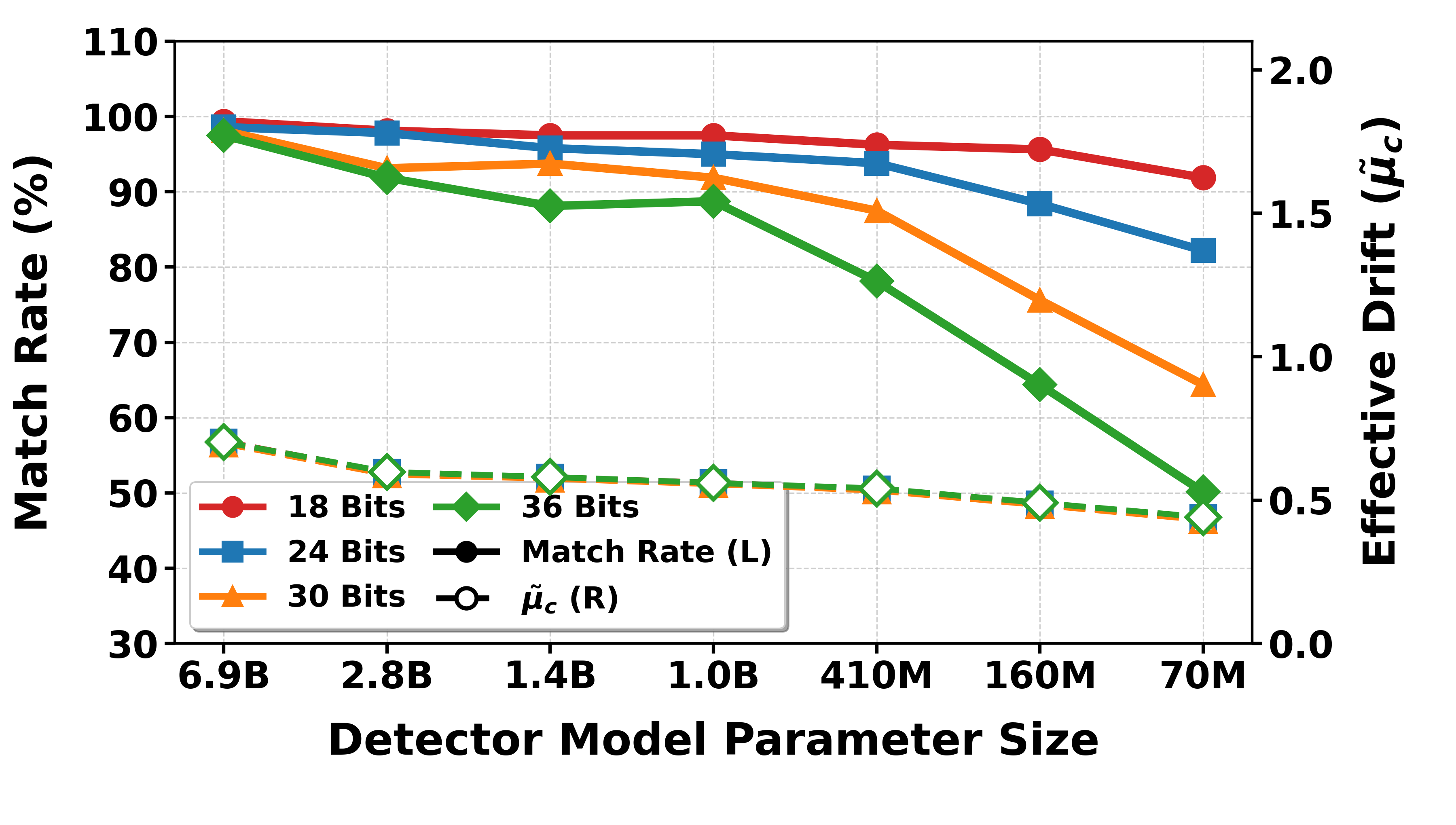}
        \caption{Verifier size: accuracy vs.\ latency.}
        \label{fig:model_ablation}
    \end{subfigure}
    \hfill
    \begin{subfigure}[t]{0.49\textwidth}
        \centering
        \includegraphics[width=\linewidth]{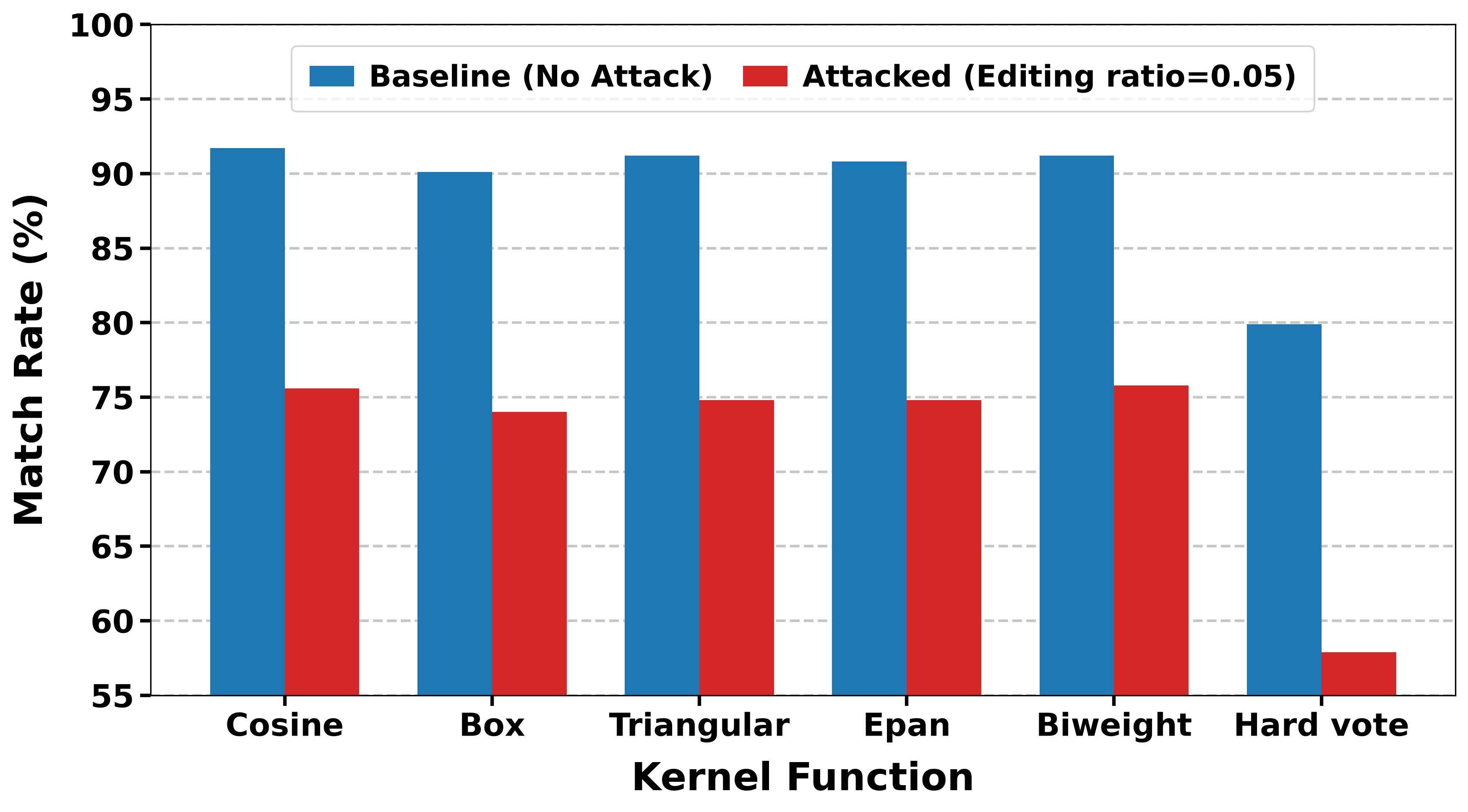}
        \caption{Kernel choice at 36-bit payload.}
        \label{fig:ker_ablation}
    \end{subfigure}

    \caption{Component ablations.}
\Description{Four line-chart panels show that Gray coding and SPC improve match rate, reducing verifier size mainly hurts longer payloads, and smooth kernels perform similarly while hard voting is markedly worse.}
    \label{fig:comprehensive_ablation}
\end{figure*}

We quantify the contribution of each design choice by isolating the main components of the method.

\noindent\textbf{Discriminative power of soft confidence.}
Figure~\ref{fig:llr_density} examines the LLR density. Without attack, correct bits concentrate in regions with large-magnitude LLRs. Under 10\% token-level editing, the distributions move closer to zero overall, and incorrect bits concentrate in the low-confidence region. Even so, the AUC remains 0.952, which shows that LLR magnitude still has strong discriminative power after editing. This supports the SPC-repair design: low-$|\Lambda|$ positions are flipped first, mimicking soft-decision erasure decoding without explicitly discarding evidence.

\noindent\textbf{Effect of Gray coding.}
Figure~\ref{fig:gray_code_ablation} shows that as the payload increases, Gray coding maintains higher match rate and bit accuracy than standard binary indexing. This aligns with the design intuition in Section~\ref{sec:method}: an adjacent-sector error changes only one Gray-code bit.

\noindent\textbf{Error correction: SPC versus classical codes.}
Figure~\ref{fig:ecc_ablation} compares SPC with Hamming-(7,4) and Reed--Solomon RS-(6,4). Classical hard-decision codes can correct at most one erroneous unit but require higher redundancy. RS-Soft extends recovery to two symbols by masking the least reliable locations as erasures, but it introduces extra symbol-level overhead. Because Gray mapping makes errors sparse and typically reduces them to single-bit flips, SPC can achieve strong recovery with very low redundancy by flipping only the bit with the smallest $|\Lambda|$, which matches the soft-decision pipeline particularly well.

\noindent\textbf{Verifier model size in the asymmetric setting.}
Figure~\ref{fig:model_ablation} shows that compact proxy models still maintain a high match rate in low-payload settings of 18--24 bits, which makes verification feasible under tight resource budgets. As the payload increases, the gap between proxy models widens. Proxy models with at least 410M parameters remain stable across payload lengths and are a practical choice for low-resource verifiers. The effective drift $\mu_c$ varies mainly with verifier size and is largely insensitive to payload length. Even under substantial distribution mismatch, $\mu_c$ remains well above zero, around 0.5, satisfying the positive-margin condition used in Proposition~\ref{prop:azuma_tool_tm0}.

\noindent\textbf{Kernel choice for interval evidence.}
Figure~\ref{fig:ker_ablation} compares different scoring rules for sector compatibility. Match rates remain stable across finite-support kernels, including cosine, box, triangular, Epanechnikov(Epan), and biweight kernels. Hard voting, by contrast, degrades sharply under the same perturbations. This highlights the value of soft interval evidence. Kernel integration uses interval geometry to produce graded sector weights instead of reducing each observation to a single nearest-anchor decision. The cosine kernel performs strongly and consistently across settings, so we use it as the default. Formal kernel definitions appear in Appendix~\ref{app:kernels}.

Taken together, Gray coding, soft interval evidence, and SPC play complementary roles. Gray coding localizes bit errors, soft interval evidence preserves reliable signal under mismatch and editing, and SPC provides confidence-guided correction with minimal redundancy. Their interaction is the main reason PURA achieves robust high-capacity extraction.

\section{Conclusion}

We present PURA, an unbiased multi-bit watermarking framework that
encodes payload bits as phase rotations of the latent sampling point
and recovers them by treating each observed token as interval
evidence and aggregating it across the sequence into per-bit
confidence. Experiments show that PURA preserves the generation
distribution and stealth while substantially exceeding existing
baselines in both capacity and recovery accuracy, and remains stable
under a broad range of attacks.

On the theoretical side, we provide the first robustness analysis
framework for multi-bit watermarking that covers adaptive attackers
with knowledge of the watermarking mechanism, yielding a per-bit
error bound that decays exponentially with sequence length.

Combined with millisecond-level asymmetric verification, PURA offers
a practical path toward fine-grained attribution in real-world
deployment.

\section{Limitations and Future Work}
Like existing context-based watermarking schemes, PURA derives keyed
randomness from local contexts. Reusing the same key across responses
with repeated contexts may therefore introduce cross-response
correlations, a common limitation of such schemes. Extending the
current single-invocation guarantee to multiple responses remains
future work.

Recovery also depends on accumulating sufficient usable evidence.
Short texts and extensive rewriting can reduce exact-message recovery,
so the remaining bit-level evidence can support preliminary
auditing and candidate narrowing rather than exact payload recovery.
Our asymmetric evaluation uses source--proxy pairs that share the same
tokenizer and vocabulary and focuses on English text. Future work will
examine cross-tokenizer verification and broader multilingual settings.

\FloatBarrier
\bibliographystyle{ACM-Reference-Format}
\bibliography{ref}

@misc{gilani2026arcmarkdistortionfreemultibytellm,
      title={ArcMark: Distortion-Free Multi-Byte LLM Watermark via Optimal Transport}, 
      author={Atefeh Gilani and Sajani Vithana and Carol Xuan Long and Oliver Kosut and Lalitha Sankar and Flavio P. Calmon},
      year={2026},
      eprint={2602.07235},
      archivePrefix={arXiv},
      primaryClass={cs.LG},
      url={https://arxiv.org/abs/2602.07235}, 
}

@misc{cui2026mc2markdistortionfreemultibitwatermarking,
      title={MC2{$^2$}Mark: Distortion-Free Multi-Bit Watermarking for Long Messages}, 
      author={Xuehao Cui and Ruibo Chen and Yihan Wu and Heng Huang},
      year={2026},
      eprint={2602.14030},
      archivePrefix={arXiv},
      primaryClass={cs.CR},
      url={https://arxiv.org/abs/2602.14030}, 
}

@inproceedings{
zhao2024provable,
title={Provable Robust Watermarking for {AI}-Generated Text},
author={Xuandong Zhao and Prabhanjan Vijendra Ananth and Lei Li and Yu-Xiang Wang},
booktitle={The Twelfth International Conference on Learning Representations},
year={2024},
url={https://openreview.net/forum?id=SsmT8aO45L}
}

@inproceedings{DBLP:conf/iclr/WangYC0LM0024,
  author       = {Lean Wang and
                  Wenkai Yang and
                  Deli Chen and
                  Hao Zhou and
                  Yankai Lin and
                  Fandong Meng and
                  Jie Zhou and
                  Xu Sun},
  title        = {Towards Codable Watermarking for Injecting Multi-Bits Information
                  to LLMs},
  booktitle    = {The Twelfth International Conference on Learning Representations,
                  {ICLR} 2024, Vienna, Austria, May 7-11, 2024},
  publisher    = {OpenReview.net},
  year         = {2024},
  url          = {https://openreview.net/forum?id=JYu5Flqm9D},
  bibsource    = {dblp computer science bibliography, https://dblp.org}
}

@article{qwen2,
      title={Qwen2 Technical Report}, 
      author={An Yang and Baosong Yang and Binyuan Hui and Bo Zheng and Bowen Yu and Chang Zhou and Chengpeng Li and Chengyuan Li and Dayiheng Liu and Fei Huang and Guanting Dong and Haoran Wei and Huan Lin and Jialong Tang and Jialin Wang and Jian Yang and Jianhong Tu and Jianwei Zhang and Jianxin Ma and Jin Xu and Jingren Zhou and Jinze Bai and Jinzheng He and Junyang Lin and Kai Dang and Keming Lu and Keqin Chen and Kexin Yang and Mei Li and Mingfeng Xue and Na Ni and Pei Zhang and Peng Wang and Ru Peng and Rui Men and Ruize Gao and Runji Lin and Shijie Wang and Shuai Bai and Sinan Tan and Tianhang Zhu and Tianhao Li and Tianyu Liu and Wenbin Ge and Xiaodong Deng and Xiaohuan Zhou and Xingzhang Ren and Xinyu Zhang and Xipin Wei and Xuancheng Ren and Yang Fan and Yang Yao and Yichang Zhang and Yu Wan and Yunfei Chu and Yuqiong Liu and Zeyu Cui and Zhenru Zhang and Zhihao Fan},
      journal={arXiv preprint arXiv:2407.10671},
      year={2024}
}

@book{roch_mdp_2024,
place={Cambridge},
series={Cambridge Series in Statistical and Probabilistic Mathematics},
title={Modern Discrete Probability: An Essential Toolkit},
DOI={10.1017/9781009305129},
publisher={Cambridge University Press},
author={Roch, Sebastien},
year={2024},
collection={Cambridge Series in Statistical and Probabilistic Mathematics}
}

@ARTICLE{11329500,
  author={Jiang, Jun and Chen, Kejiang and Zhao, Na and Qi, Yuang and Zhang, Xin and Zhang, Weiming and Yu, Nenghai},
  journal={IEEE Transactions on Multimedia}, 
  title={Disreo: Provably Secure No-Box-Extraction Linguistic Steganography Based on Distribution Reorganization}, 
  year={2026},
  volume={28},
  number={},
  pages={2970-2983},
  doi={10.1109/TMM.2026.3651038}}

@InProceedings{pmlr-v267-cheng25c,
  title = 	 {Revealing Weaknesses in Text Watermarking Through Self-Information Rewrite Attacks},
  author =       {Cheng, Yixin and Guo, Hongcheng and Li, Yangming and Sigal, Leonid},
  booktitle = 	 {Proceedings of the 42nd International Conference on Machine Learning},
  pages = 	 {9982--10009},
  year = 	 {2025},
  editor = 	 {Singh, Aarti and Fazel, Maryam and Hsu, Daniel and Lacoste-Julien, Simon and Berkenkamp, Felix and Maharaj, Tegan and Wagstaff, Kiri and Zhu, Jerry},
  volume = 	 {267},
  series = 	 {Proceedings of Machine Learning Research},
  month = 	 {13--19 Jul},
  publisher =    {PMLR},
  url = 	 {https://proceedings.mlr.press/v267/cheng25c.html}
}

@InProceedings{pmlr-v235-jovanovic24a,
  title = 	 {Watermark Stealing in Large Language Models},
  author =       {Jovanovi\'{c}, Nikola and Staab, Robin and Vechev, Martin},
  booktitle = 	 {Proceedings of the 41st International Conference on Machine Learning},
  pages = 	 {22570--22593},
  year = 	 {2024},
  editor = 	 {Salakhutdinov, Ruslan and Kolter, Zico and Heller, Katherine and Weller, Adrian and Oliver, Nuria and Scarlett, Jonathan and Berkenkamp, Felix},
  volume = 	 {235},
  series = 	 {Proceedings of Machine Learning Research},
  month = 	 {21--27 Jul},
  publisher =    {PMLR},
  url = 	 {https://proceedings.mlr.press/v235/jovanovic24a.html}
}

@inproceedings{chang-etal-2025-watermark,
    title = "Watermark Smoothing Attacks against Language Models",
    author = "Chang, Hongyan  and
      Hassani, Hamed  and
      Shokri, Reza",
    editor = "Christodoulopoulos, Christos  and
      Chakraborty, Tanmoy  and
      Rose, Carolyn  and
      Peng, Violet",
    booktitle = "Findings of the Association for Computational Linguistics: EMNLP 2025",
    month = nov,
    year = "2025",
    address = "Suzhou, China",
    publisher = "Association for Computational Linguistics",
    url = "https://aclanthology.org/2025.findings-emnlp.264/",
    doi = "10.18653/v1/2025.findings-emnlp.264",
    pages = "4915--4941",
    ISBN = "979-8-89176-335-7"
}

@inproceedings{diaa2025optimizing,
  title={Optimizing Adaptive Attacks against Watermarks for Language Models},
  author={Abdulrahman Diaa and Toluwani Aremu and Nils Lukas},
  booktitle={Forty-second International Conference on Machine Learning},
  year={2025},
  url={https://openreview.net/forum?id=AsODat0dkE}
}

@inproceedings{chen2025demark,
  title={{DE-MARK}: Watermark Removal in Large Language Models},
  author={Chen, Ruibo and Wu, Yihan and Guo, Junfeng and Huang, Heng},
  booktitle={Proceedings of the 42nd International Conference on Machine Learning (ICML)},
  year={2025}
}

@article{huang2025rlcracker,
  title={{RLCracker}: Exposing the Vulnerability of LLM Watermarks with Adaptive RL Attacks},
  author={Huang, Hanbo and Zhang, Yiran and Zheng, Hao and Gong, Xuan and Li, Yihan and Liu, Lin and Liang, Shiyu},
  journal={arXiv preprint},
  year={2025}
}

@inproceedings{reynolds2025breaking,
  title={Breaking Distortion-free Watermarks in Large Language Models},
  author={Shayleen Reynolds and Hengzhi He and Dung Daniel Ngo and Saheed Obitayo and Niccolo Dalmasso and Guang Cheng and Vamsi K. Potluru and Manuela Veloso},
  booktitle={Lock-LLM Workshop: Prevent Unauthorized Knowledge Use from Large Language Models},
  year={2025},
  url={https://openreview.net/forum?id=six52ViBmy}
}

@inproceedings{wang2026ANStega,
  title={Breaking the Generative Steganography Trilemma: ANStega for Optimal Capacity, Efficiency, and Security},
  author={Wang, Yaofei and Pang, Weilong and Chen, Kejiang and Ding, Jinyang and Zhang, W and Hu, D and Yu, Nenghai},
  booktitle={33rd Annual Network and Distributed System Security Symposium (NDSS’26), The Internet Society},
  year={2026}
}

@inproceedings{zhang2026character,
  title={Character-Level Perturbations Disrupt {LLM} Watermarks},
  author={Zhang, Zhaoxi and others}, 
  booktitle={Proceedings of the Network and Distributed System Security Symposium (NDSS)},
  year={2026}
}

@article{yang2020tscsw,
  title   = {TS-CSW: text steganalysis and hidden capacity estimation based on convolutional sliding windows},
  author  = {Yang, Zhongliang and Huang, Yongfeng and Zhang, Yu-Jin},
  journal = {Multimedia Tools and Applications},
  volume  = {79},
  pages   = {18293--18316},
  year    = {2020},
  doi     = {10.1007/s11042-020-08716-w}
}

@INPROCEEDINGS{Ding2023Discop,
  author={Ding, Jinyang and Chen, Kejiang and Wang, Yaofei and Zhao, Na and Zhang, Weiming and Yu, Nenghai},
  booktitle={2023 IEEE Symposium on Security and Privacy (SP)}, 
  title={Discop: Provably Secure Steganography in Practice Based on "Distribution Copies"}, 
  year={2023},
  volume={},
  number={},
  pages={2238-2255},
  doi={10.1109/SP46215.2023.10179287}}

@inproceedings{iMEC,
title={Perfectly Secure Steganography Using Minimum Entropy Coupling},
author={Christian Schroeder de Witt and Samuel Sokota and J Zico Kolter and Jakob Nicolaus Foerster and Martin Strohmeier},
booktitle={The Eleventh International Conference on Learning Representations },
year={2023},
url={https://openreview.net/forum?id=HQ67mj5rJdR}
}

@inproceedings{Wang2025SparSamp,
  author = {Yaofei Wang and Gang Pei and Kejiang Chen and Jinyang Ding and Chao Pan and Weilong Pang and Donghui Hu and Weiming Zhang},
  title = {{SparSamp}: Efficient Provably Secure Steganography Based on Sparse Sampling},
  booktitle = {34th USENIX Security Symposium (USENIX Security 25)},
  year = {2025},
  isbn = {978-1-939133-52-6},
  address = {Seattle, WA},
  pages = {6817--6835},
  url = {https://www.usenix.org/conference/usenixsecurity25/presentation/wang-yaofei},
  publisher = {USENIX Association},
  month = aug
}

@inproceedings{Kaptchuk2021Meteor,
author = {Kaptchuk, Gabriel and Jois, Tushar M. and Green, Matthew and Rubin, Aviel D.},
title = {Meteor: Cryptographically Secure Steganography for Realistic Distributions},
year = {2021},
isbn = {9781450384544},
publisher = {Association for Computing Machinery},
address = {New York, NY, USA},
url = {https://doi.org/10.1145/3460120.3484550},
doi = {10.1145/3460120.3484550},
booktitle = {Proceedings of the 2021 ACM SIGSAC Conference on Computer and Communications Security},
pages = {1529–1548},
numpages = {20},
location = {Virtual Event, Republic of Korea},
series = {CCS '21}
}

@inproceedings{Ziegler2019Neural,
    title = "Neural Linguistic Steganography",
    author = "Ziegler, Zachary  and
      Deng, Yuntian  and
      Rush, Alexander",
    editor = "Inui, Kentaro  and
      Jiang, Jing  and
      Ng, Vincent  and
      Wan, Xiaojun",
    booktitle = "Proceedings of the 2019 Conference on Empirical Methods in Natural Language Processing and the 9th International Joint Conference on Natural Language Processing (EMNLP-IJCNLP)",
    month = nov,
    year = "2019",
    address = "Hong Kong, China",
    publisher = "Association for Computational Linguistics",
    url = "https://aclanthology.org/D19-1115/",
    doi = "10.18653/v1/D19-1115",
    pages = "1210--1215"}

@inproceedings{bai_provably_2025,
	address = {San Francisco, CA, USA},
	title = {Provably {Robust} and {Secure} {Steganography} in {Asymmetric} {Resource} {Scenario}},
	copyright = {https://doi.org/10.15223/policy-029},
	isbn = {979-8-3315-2236-0},
	url = {https://ieeexplore.ieee.org/document/11023508/},
	doi = {10.1109/SP61157.2025.00155},
	language = {en},
	urldate = {2025-09-24},
	booktitle = {2025 {IEEE} {Symposium} on {Security} and {Privacy} ({SP})},
	publisher = {IEEE},
	author = {Bai, Minhao and Yang, Jinshuai and Pang, Kaiyi and Xu, Xin and Yang, Zhen and Huang, Yongfeng},
	month = may,
	year = {2025},
	pages = {1438--1456},
}

@misc{ivypanda2024,
  author       = {{IvyPanda}},
  title        = {{IvyPanda}},
  howpublished = {\url{https://ivypanda.com/essays/}},
  year         = {2024},
  note         = {Accessed: January 19, 2024}
}

@article{rael_exploring_nodate,
	title = {Exploring the {Limits} of {Transfer} {Learning} with a {Uniﬁed} {Text}-to-{Text} {Transformer}},
	language = {en},
	author = {Raﬀel, Colin and Shazeer, Noam and Roberts, Adam and Lee, Katherine and Narang, Sharan and Matena, Michael and Zhou, Yanqi and Li, Wei and Liu, Peter J},
}

@misc{merity_pointer_2016,
	title = {Pointer {Sentinel} {Mixture} {Models}},
	url = {http://arxiv.org/abs/1609.07843},
	doi = {10.48550/arXiv.1609.07843},
	language = {en},
	urldate = {2026-01-05},
	publisher = {arXiv},
	author = {Merity, Stephen and Xiong, Caiming and Bradbury, James and Socher, Richard},
	month = sep,
	year = {2016},
	note = {arXiv:1609.07843 [cs]},
}

@InProceedings{kirchenbauer2023KGW,
  title = 	 {A Watermark for Large Language Models},
  author =       {Kirchenbauer, John and Geiping, Jonas and Wen, Yuxin and Katz, Jonathan and Miers, Ian and Goldstein, Tom},
  booktitle = 	 {Proceedings of the 40th International Conference on Machine Learning},
  pages = 	 {17061--17084},
  year = 	 {2023},
  editor = 	 {Krause, Andreas and Brunskill, Emma and Cho, Kyunghyun and Engelhardt, Barbara and Sabato, Sivan and Scarlett, Jonathan},
  volume = 	 {202},
  series = 	 {Proceedings of Machine Learning Research},
  month = 	 {23--29 Jul},
  publisher =    {PMLR},
  url = 	 {https://proceedings.mlr.press/v202/kirchenbauer23a.html}
}

@inproceedings{ICLR2024_c5b00c5b,
 author = {Hu, Zhengmian and Chen, Lichang and Wu, Xidong and Wu, Yihan and Zhang, Hongyang and Huang, Heng},
 booktitle = {International Conference on Learning Representations},
 editor = {B. Kim and Y. Yue and S. Chaudhuri and K. Fragkiadaki and M. Khan and Y. Sun},
 pages = {45408--45436},
 title = {Unbiased Watermark for Large Language Models},
 url = {https://proceedings.iclr.cc/paper_files/paper/2024/file/c5b00c5bdcc6fe35907dbcca03d27652-Paper-Conference.pdf},
 volume = {2024},
 year = {2024}
}

@article{kuditipudi2023its,
title={Robust Distortion-free Watermarks for Language Models},
author={Rohith Kuditipudi and John Thickstun and Tatsunori Hashimoto and Percy Liang},
journal={Transactions on Machine Learning Research},
issn={2835-8856},
year={2024},
url={https://openreview.net/forum?id=FpaCL1MO2C},
note={}
}

@INPROCEEDINGS{fernandez2023cycleshift,
  author={Fernandez, Pierre and Chaffin, Antoine and Tit, Karim and Chappelier, Vivien and Furon, Teddy},
  booktitle={2023 IEEE International Workshop on Information Forensics and Security (WIFS)}, 
  title={Three Bricks to Consolidate Watermarks for Large Language Models}, 
  year={2023},
  volume={},
  number={},
  pages={1-6},
  doi={10.1109/WIFS58808.2023.10374576}}

@inproceedings{yoo2023MPAC,
    title = "Advancing Beyond Identification: Multi-bit Watermark for Large Language Models",
    author = "Yoo, KiYoon  and
      Ahn, Wonhyuk  and
      Kwak, Nojun",
    editor = "Duh, Kevin  and
      Gomez, Helena  and
      Bethard, Steven",
    booktitle = "Proceedings of the 2024 Conference of the North American Chapter of the Association for Computational Linguistics: Human Language Technologies (Volume 1: Long Papers)",
    month = jun,
    year = "2024",
    address = "Mexico City, Mexico",
    publisher = "Association for Computational Linguistics",
    url = "https://aclanthology.org/2024.naacl-long.224/",
    doi = "10.18653/v1/2024.naacl-long.224",
    pages = "4031--4055"
}

@InProceedings{feng2025bimark,
  title = 	 {{B}i{M}ark: Unbiased Multilayer Watermarking for Large Language Models},
  author =       {Feng, Xiaoyan and Zhang, He and Zhang, Yanjun and Zhang, Leo Yu and Pan, Shirui},
  booktitle = 	 {Proceedings of the 42nd International Conference on Machine Learning},
  pages = 	 {17049--17067},
  year = 	 {2025},
  editor = 	 {Singh, Aarti and Fazel, Maryam and Hsu, Daniel and Lacoste-Julien, Simon and Berkenkamp, Felix and Maharaj, Tegan and Wagstaff, Kiri and Zhu, Jerry},
  volume = 	 {267},
  series = 	 {Proceedings of Machine Learning Research},
  month = 	 {13--19 Jul},
  publisher =    {PMLR},
  url = 	 {https://proceedings.mlr.press/v267/feng25u.html}
}

@InProceedings{jiang2025stealthink,
  title = 	 {{S}tealth{I}nk: A Multi-bit and Stealthy Watermark for Large Language Models},
  author =       {Jiang, Ya and Wu, Chuxiong and Boroujeny, Massieh Kordi and Mark, Brian and Zeng, Kai},
  booktitle = 	 {Proceedings of the 42nd International Conference on Machine Learning},
  pages = 	 {27685--27709},
  year = 	 {2025},
  editor = 	 {Singh, Aarti and Fazel, Maryam and Hsu, Daniel and Lacoste-Julien, Simon and Berkenkamp, Felix and Maharaj, Tegan and Wagstaff, Kiri and Zhu, Jerry},
  volume = 	 {267},
  series = 	 {Proceedings of Machine Learning Research},
  month = 	 {13--19 Jul},
  publisher =    {PMLR},
  url = 	 {https://proceedings.mlr.press/v267/jiang25j.html}
}

@misc{aaronson2023openai,
  title={`{R}eform' {AI} {A}lignment with {S}cott {A}aronson},
  author={Scott Aaronson},
  howpublished={AXRP - the AI X-risk Research Podcast},
  year={2023},
  url={https://axrp.net/episode/2023/04/11/episode-20-reform-ai-alignment-scott-aaronson.html}
}

@inproceedings{qu2024via_code,
author = {Qu, Wenjie and Zheng, Wengrui and Tao, Tianyang and Yin, Dong and Jiang, Yanze and Tian, Zhihua and Zou, Wei and Jia, Jinyuan and Zhang, Jiaheng},
title = {Provably robust multi-bit watermarking for AI-generated text},
year = {2025},
isbn = {978-1-939133-52-6},
publisher = {USENIX Association},
address = {USA},
booktitle = {Proceedings of the 34th USENIX Conference on Security Symposium},
articleno = {11},
numpages = {20},
location = {Seattle, WA, USA},
series = {SEC '25}
}

@InProceedings{wu2023dipmark,
  title = 	 {A Resilient and Accessible Distribution-Preserving Watermark for Large Language Models},
  author =       {Wu, Yihan and Hu, Zhengmian and Guo, Junfeng and Zhang, Hongyang and Huang, Heng},
  booktitle = 	 {Proceedings of the 41st International Conference on Machine Learning},
  pages = 	 {53443--53470},
  year = 	 {2024},
  editor = 	 {Salakhutdinov, Ruslan and Kolter, Zico and Heller, Katherine and Weller, Adrian and Oliver, Nuria and Scarlett, Jonathan and Berkenkamp, Felix},
  volume = 	 {235},
  series = 	 {Proceedings of Machine Learning Research},
  month = 	 {21--27 Jul},
  publisher =    {PMLR},
  url = 	 {https://proceedings.mlr.press/v235/wu24h.html}
}

@misc{achiam2023gpt4,
      title={GPT-4 Technical Report}, 
      author={OpenAI and Josh Achiam and Steven Adler and others},
      year={2024},
      eprint={2303.08774},
      archivePrefix={arXiv},
      primaryClass={cs.CL},
      url={https://arxiv.org/abs/2303.08774}, 
}

@article{biden2023executive,
  title={Executive order on the safe, secure, and trustworthy development and use of artificial intelligence},
  author={Biden, Joseph R},
  year={2023}
}

@article{smuha2025regulation,
  title={Regulation 2024/1689 of the Eur. Parl. \& Council of June 13, 2024 (EU Artificial Intelligence Act)},
  author={Smuha, Nathalie A},
  journal={International Legal Materials},
  pages={1--148},
  year={2025},
  publisher={Cambridge University Press}
}

@inproceedings{mao2024STA,
    title = "Watermarking Large Language Models: An Unbiased and Low-risk Method",
    author = "Mao, Minjia  and
      Wei, Dongjun  and
      Chen, Zeyu  and
      Fang, Xiao  and
      Chau, Michael",
    editor = "Che, Wanxiang  and
      Nabende, Joyce  and
      Shutova, Ekaterina  and
      Pilehvar, Mohammad Taher",
    booktitle = "Proceedings of the 63rd Annual Meeting of the Association for Computational Linguistics (Volume 1: Long Papers)",
    month = jul,
    year = "2025",
    address = "Vienna, Austria",
    publisher = "Association for Computational Linguistics",
    url = "https://aclanthology.org/2025.acl-long.391/",
    doi = "10.18653/v1/2025.acl-long.391",
    pages = "7939--7960",
    ISBN = "979-8-89176-251-0"
}

@inproceedings{mcmark2025Chen,
    title = "Improved Unbiased Watermark for Large Language Models",
    author = "Chen, Ruibo  and
      Wu, Yihan  and
      Guo, Junfeng  and
      Huang, Heng",
    editor = "Che, Wanxiang  and
      Nabende, Joyce  and
      Shutova, Ekaterina  and
      Pilehvar, Mohammad Taher",
    booktitle = "Proceedings of the 63rd Annual Meeting of the Association for Computational Linguistics (Volume 1: Long Papers)",
    month = jul,
    year = "2025",
    address = "Vienna, Austria",
    publisher = "Association for Computational Linguistics",
    url = "https://aclanthology.org/2025.acl-long.1005/",
    doi = "10.18653/v1/2025.acl-long.1005",
    pages = "20587--20601",
    ISBN = "979-8-89176-251-0"
}

@article{krishna2023paraphrasing,
  title={Paraphrasing evades detectors of ai-generated text, but retrieval is an effective defense},
  author={Krishna, Kalpesh and Song, Yixiao and Karpinska, Marzena and Wieting, John and Iyyer, Mohit},
  journal={Advances in Neural Information Processing Systems},
  volume={36},
  pages={27469--27500},
  year={2023}
}

@misc{touvron2023llama,
      title={Llama 2: Open Foundation and Fine-Tuned Chat Models}, 
      author={Hugo Touvron and Louis Martin and Kevin Stone and others},
      year={2023},
      eprint={2307.09288},
      archivePrefix={arXiv},
      primaryClass={cs.CL},
      url={https://arxiv.org/abs/2307.09288}, 
}

@article{dathathri2024synthid,
  title={Scalable watermarking for identifying large language model outputs},
  author={Dathathri, Sumanth and See, Abigail and Ghaisas, Sumedh and Huang, Po-Sen and McAdam, Rob and Welbl, Johannes and Bachani, Vandana and Kaskasoli, Alex and Stanforth, Robert and Matejovicova, Tatiana and others},
  journal={Nature},
  volume={634},
  number={8035},
  pages={818--823},
  year={2024},
  publisher={Nature Publishing Group UK London}
}

@misc{grattafiori_llama_2024,
	title = {The {Llama} 3 {Herd} of {Models}},
	url = {http://arxiv.org/abs/2407.21783},
	doi = {10.48550/arXiv.2407.21783},
	language = {en},
	urldate = {2026-01-12},
	publisher = {arXiv},
	author = {Grattafiori, Aaron and Dubey, Abhimanyu and Jauhri, Abhinav and others},
	month = nov,
	year = {2024},
	note = {arXiv:2407.21783 [cs]},
}

@inproceedings{pythia_icml_2023,
  title     = {Pythia: A Suite for Analyzing Large Language Models Across Training and Scaling},
  author    = {Biderman, Stella and Schoelkopf, Hailey and others},
  booktitle = {Proceedings of the 40th International Conference on Machine Learning (ICML)},
  year      = {2023}
}

@article{bloom_2022,
  title   = {BLOOM: A 176B-Parameter Open-Access Multilingual Language Model},
  author  = {BigScience Workshop},
  journal = {arXiv preprint arXiv:2211.05100},
  year    = {2022}
}

@article{gemma2_2024,
  title   = {Gemma 2: Improving Open Language Models at a Practical Size},
  author  = {Gemma Team},
  journal = {arXiv preprint arXiv:2408.00118},
  year    = {2024}
}
\appendix
\section{Ethical Considerations}
\label{app:ethics}

\noindent\textbf{Scope.}
We study robust multi-bit watermarking and asymmetric verification for LLM-generated text. The relevant stakeholders include end users, content creators, model and platform providers, auditors, and adversaries.

\noindent\textbf{Benefits.}
Provenance signals can support accountability and incident response while preserving text utility by avoiding distortion-inducing bias.

\noindent\textbf{Risks.}
Watermarking may be misused for surveillance or censorship. Detection is probabilistic and may be misinterpreted as conclusive evidence. Publicly disclosing robust designs may also stimulate stronger removal or forgery attacks. Performance may additionally vary across domains and languages.

\noindent\textbf{Mitigation and responsible use.}
We recommend using payloads that do not encode personally identifiable information by default, together with strict key management. In high-risk settings, evidence should be used conservatively, with calibrated thresholds and explicit error-rate reporting. Assumptions, threat models, and measured performance should be reported transparently. We caution against automated enforcement or punitive decisions based solely on watermark detection.
\section{Computational Unbiasedness of PURA for a Single Invocation}
\label{app:onecall_unbiased}

\subsection{Setup}
\label{app:setup_its}

We follow the notation from Section~\ref{subsec:its_background}. For an autoregressive language model over vocabulary $\mathcal{V}$, the baseline generator $\mathsf{Gen}_{\mathrm{Baseline}}$ samples $x_t \sim \pi_t(\cdot)$ at each step $t$. Fix the secret permuted order induced by $\Pi$ as
\[
\tilde{\mathcal{V}}=(\tilde{v}_1,\ldots,\tilde{v}_{|\mathcal{V}|}),
\]
and define the corresponding CDF
\begin{equation}
\tilde{F}_t(k) = \sum_{j=1}^{k}\pi_t(\tilde{v}_j),
\qquad \tilde{F}_t(0)=0.
\label{eq:app_Ft}
\end{equation}
Sampling $u_t\sim \mathrm{Unif}[0,1)$ and setting
\begin{equation}
\tilde{k}_t = \min\{k \mid \tilde{F}_t(k) > u_t\},
\qquad x_t = \tilde{v}_{\tilde{k}_t}
\label{eq:app_baseline_its}
\end{equation}
produces $\pi_t(\cdot)$ exactly, regardless of the order.

\subsection{Latent-Space Construction of PURA}
\label{app:setup_pura}

Let a keyed hash output be $h\in\{0,1\}^\lambda$, and define $\mathrm{Unif}(h)=\mathrm{Int}(h)/2^\lambda\in[0,1)$. We use the standard idealization that this conversion is uniform on $[0,1)$, although finite precision introduces only negligible discretization error. The keyed derivations used for seed generation, payload indexing, context gating, and the vocabulary permutation are separated by public tag prefixes and analyzed through the corresponding random-function hybrid under a pseudorandom-function (PRF) assumption.

Let $\mathsf{ctx}_t=x_{t-W:t-1}$. At a watermarked step $t$, PURA computes the base phase $U_t$ and the payload offset $\Delta_t$ according to Equations~\eqref{eq:base_seed}--\eqref{eq:anchor}, and sets
\begin{equation}
u'_t \triangleq U_t \oplus \Delta_t,
\label{eq:app_pura_vars}
\end{equation}
where $\Delta_t=(s_t+0.5)/K_t$, $K_t=2^{a_t}$, and $s_t\in\{0,\ldots,K_t-1\}$. In the random-function hybrid, domain separation makes $\Delta_t$ independent of the seed hash output used to form $U_t$. The uniqueness constraint in Equation~\eqref{eq:uniqueness} ensures that all context inputs $\mathsf{ctx}_t$ at watermarked steps are distinct within one session, up to negligible hash collisions. Skip steps sample directly from $\pi_t(\cdot)$, while watermarked steps apply ITS with $u'_t$ in the permuted space:
\[
\tilde{k}_t = \min\{k \mid \tilde{F}_t(k) > u'_t\},
\qquad x_t=\tilde{v}_{\tilde{k}_t}.
\]

\subsection{Independence Lemmas}
\label{app:independence_analysis}

\begin{lemma}[Rotational invariance]
\label{lem:app_rotation_invariance}
If $u\sim \mathrm{Unif}[0,1)$ and $\Delta\in[0,1)$ is independent of $u$, then $u\oplus\Delta\sim\mathrm{Unif}[0,1)$.
\end{lemma}

\begin{proof}
The map $u\mapsto u\oplus\Delta$ is a measure-preserving bijection on $\mathbb{R}/\mathbb{Z}$.
\end{proof}

\begin{lemma}[Independent uniformity under distinct queries]
\label{lem:app_dedup_independence}
Let $f:\mathcal{C}\to[0,1)$ be a truly random function. If an adaptive process queries distinct contexts $\mathsf{ctx}_t$ during one execution, then the returned values are independent and uniformly distributed on $[0,1)$.
\end{lemma}

\begin{proof}
A truly random function returns independent uniform values on distinct inputs, regardless of whether the queries are adaptive.
\end{proof}

\begin{corollary}[Fresh base phase]
\label{cor:fresh_phase_uniform}
In the random-function hybrid and under the uniqueness constraint of
PURA, for every watermarked step $t$,
\[
U_t \mid \mathcal{G}_{t-1} \sim \mathrm{Unif}[0,1).
\]
\end{corollary}

\begin{proof}
The uniqueness constraint guarantees that $\mathsf{ctx}_t$ is a fresh input at every step, so the conclusion follows immediately from Lemma~\ref{lem:app_dedup_independence}.
\end{proof}

\subsection{Proof of Computational Unbiasedness}
\label{app:comp_unbiased_proof}

We model the domain-separated keyed derivations as secure pseudorandom functions and analyze them through the corresponding random-function hybrid.

\begin{theorem}[Computational unbiasedness for a single invocation]
\label{thm:app_onecall_comp}
Under the repeated-context skipping rule in Equation~\eqref{eq:uniqueness} and the PRF assumption, for any PPT distinguisher $\mathcal{D}$,
\[
\Bigl|\Pr[\mathcal{D}(\mathsf{Gen}_{\mathrm{PURA}})=1]
-\Pr[\mathcal{D}(\mathsf{Gen}_{\mathrm{Baseline}})=1]\Bigr|
\le \mathrm{negl}(\lambda).
\]
\end{theorem}

\begin{proof}
We construct two hybrid distributions.

\noindent
\textbf{H0} is the real PURA process. Watermarked steps use ITS with $u'_t=U_t\oplus\Delta_t$, and skip steps sample from $\pi_t(\cdot)$ exactly as in Section~\ref{sec:method}.

\noindent
\textbf{H1} replaces the domain-separated keyed derivations with independent random-function outputs on their prefixed inputs.

\noindent\textbf{Step 1: H0 $\approx$ H1.}
By PRF security, any algorithm that distinguishes H0 from H1 yields a PRF distinguisher, so
\[
|\Pr[\mathcal{D}(\mathrm{H0})=1]
-\Pr[\mathcal{D}(\mathrm{H1})=1]|
\le \mathrm{negl}(\lambda).
\]

\noindent\textbf{Step 2: H1 $\equiv$ baseline.}
Fix a step $t$ and the history $x_{<t}$, and condition further on the permutation and the non-seed outputs used at this step. On skip steps, H1 is identical to the baseline by construction. On watermarked steps, Lemma~\ref{lem:app_dedup_independence} implies
\[
U_t\mid x_{<t}\sim\mathrm{Unif}[0,1).
\]
The offset $\Delta_t$ is then fixed, while $U_t$ remains independent of it by domain separation. Lemma~\ref{lem:app_rotation_invariance} gives
\[
u'_t=U_t\oplus\Delta_t\mid x_{<t}\sim\mathrm{Unif}[0,1).
\]
Feeding a uniform latent variable into Equation~\eqref{eq:app_baseline_its} yields $\pi_t(\cdot\mid x_{<t})$ exactly, so
\[
\Pr_{\mathrm{H1}}[x_t\in\cdot\mid x_{<t}]
=
\Pr_{\mathrm{Baseline}}[x_t\in\cdot\mid x_{<t}].
\]
By the chain rule,
\[
\Pr_{\mathrm{H1}}[x_{1:L}]
=
\Pr_{\mathrm{Baseline}}[x_{1:L}].
\]

Combining the two steps proves the theorem.
\end{proof}

\section{Proofs for Robustness}
\label{app:robustness_proofs}

\allowdisplaybreaks

\noindent\textbf{Scope.}
For the implementation evaluated in this paper, non-skipped arities satisfy $a\in\{1,2,3\}$. The symmetry arguments below are therefore verified for these three explicit Gray codebooks. The same proof pattern extends to any fixed finite arity by checking the corresponding sector permutation.

\subsection{Sector-Class Symmetry of Gray Codes}
\label{app:gray_symmetry}

For arity $a\in\{1,2,3\}$, let
\[
G_a:\{0,\dots,2^a-1\}\to\{0,1\}^a
\]
denote the explicit Gray codebook used by PURA. For each local Gray-bit coordinate $j\in\{0,\dots,a-1\}$ and each $b\in\{0,1\}$, define
\[
\mathcal{S}^{(j)}_b
\triangleq
\{\,s\in\{0,\dots,2^a-1\} : [G_a(s)]_j=b\,\}.
\]

\begin{lemma}[Sector-class symmetry of the PURA Gray codebooks]
\label{lem:gray_symmetry}
For every $a\in\{1,2,3\}$ and every coordinate $j\in\{0,\dots,a-1\}$, there exists a cyclic shift
\[
R_{a,j}(s)=s+\Delta_{a,j}\pmod{2^a}
\]
such that
\[
R_{a,j}(\mathcal{S}^{(j)}_0)=\mathcal{S}^{(j)}_1,
\qquad
R_{a,j}(\mathcal{S}^{(j)}_1)=\mathcal{S}^{(j)}_0.
\]
\end{lemma}

\begin{proof}
The claim is verified case by case over the explicit codebooks used in PURA.

When $a=1$, the codebook is $(0,1)$, so the only coordinate is swapped by a shift of $1$.

When $a=2$, the codebook is $(00,01,11,10)$. For both coordinates, a shift of $2$ swaps the corresponding $0$ and $1$ classes.

When $a=3$, the codebook is $(000,001,011,010,110,111,101,100)$. For the first two coordinates, a shift of $4$ swaps the two classes, and a shift of $2$ does so for the last coordinate.
\end{proof}

Therefore, for every valid observation with arity $a_t>0$ and every local coordinate $j$, there exists a circular rotation
\[
u \mapsto u \oplus \frac{\Delta_{a_t,j}}{2^{a_t}}
\]
that permutes the candidate anchors and swaps the two sector classes $\mathcal{S}^{(j)}_0$ and $\mathcal{S}^{(j)}_1$.

\subsection{Centering of Unsynchronized Steps Under TM0}
\label{app:tm0_centeredness}

For the TM0 analysis, we model the verifier-side phase at an
unsynchronized step as conditionally uniform given the preceding
verifier history, the synchronization outcome, and the resulting
observed interval:
\begin{equation}
\label{eq:tm0_phase_condition_app}
U_t^{\mathrm{ver}}
\mid
\mathcal{G}_{t-1},\mathsf{Val}_t=0,I_t(x'_t)
\sim \mathrm{Unif}[0,1).
\end{equation}

\begin{lemma}[Centering under TM0]
\label{lem:tm0_centeredness_app}
Under TM0, for every $t\in\mathcal{T}_i$,
\[
\mathbb{E}[\Lambda_t \mid \mathcal{G}_{t-1},\,\mathsf{Val}_t=0] = 0,
\]
and hence
\[
\mathbb{E}[Z_t \mid \mathcal{G}_{t-1},\,\mathsf{Val}_t=0] = 0.
\]
\end{lemma}

\begin{proof}
Fix $t\in\mathcal{T}_i$ and condition on
$(\mathcal{G}_{t-1},\mathsf{Val}_t=0)$ and on the observed interval
$I_t(x'_t)$. By Equation~\eqref{eq:tm0_phase_condition_app},
\[
U_t^{\mathrm{ver}}
\mid
\mathcal{G}_{t-1},\mathsf{Val}_t=0,I_t(x'_t)
\sim \mathrm{Unif}[0,1).
\]

Let $a_t' \in \{1,2,3\}$ be the arity reconstructed by the verifier
and let $K_t'=2^{a_t'}$. Since $t\in\mathcal{T}_i$, there is a unique
local Gray-bit coordinate $j\in\{0,\dots,a_t'-1\}$ such that
\[
  (\mathrm{idx}(t)+j)\bmod |c| = i .
\]
This is the local coordinate whose LLR contributes to encoded payload
position $i$. By Lemma~\ref{lem:gray_symmetry}, there exists a cyclic
shift $s \mapsto s+\Delta \pmod{K_t'}$ that swaps
$\mathcal{S}^{(j)}_0$ and $\mathcal{S}^{(j)}_1$.

The verifier constructs sector weights
\[
w_t(s)=\int_{u\in I_t(x'_t)}\mathcal{K}(u;A_{t,s}^{\mathrm{ver}})\,du,
\qquad
A_{t,s}^{\mathrm{ver}}
=
U_t^{\mathrm{ver}}\oplus \frac{s+0.5}{K_t'}.
\]
Define $\bar{U}_t^{\mathrm{ver}}=U_t^{\mathrm{ver}}\oplus \Delta/K_t'$. Then
\[
\bar{A}_{t,s}^{\mathrm{ver}}
=
\bar{U}_t^{\mathrm{ver}} \oplus \frac{s+0.5}{K_t'}
=
A_{t,s+\Delta}^{\mathrm{ver}},
\]
so $\bar{w}_t(s)=w_t(s+\Delta)$.

The local LLR contributed by this step to payload position $i$ is
\[
\Lambda_t
=
\Lambda_t^{(j)}
=
\log\frac{\sum_{s\in\mathcal{S}^{(j)}_1}w_t(s)}
          {\sum_{s\in\mathcal{S}^{(j)}_0}w_t(s)}.
\]
Because the shift swaps $\mathcal{S}^{(j)}_0$ and
$\mathcal{S}^{(j)}_1$, the numerator and denominator exchange roles
under $U_t^{\mathrm{ver}}\mapsto \bar{U}_t^{\mathrm{ver}}$, giving
\[
\Lambda_t(\bar{U}_t^{\mathrm{ver}})
=
-\Lambda_t(U_t^{\mathrm{ver}}).
\]
The clipping in Equation~\eqref{eq:local_llr} preserves this
sign reversal.

Under this conditional uniformity, the distribution of
$U_t^{\mathrm{ver}}$ is invariant under the rotation
$u\mapsto u\oplus \Delta/K_t'$. Therefore
\begin{multline*}
\mathbb{E}\bigl[\Lambda_t
\mid\mathcal{G}_{t-1},\mathsf{Val}_t=0\bigr]
=
\mathbb{E}\bigl[\Lambda_t(\bar{U}_t^{\mathrm{ver}})
\mid\mathcal{G}_{t-1},\mathsf{Val}_t=0\bigr]
\\
=
-\mathbb{E}\bigl[\Lambda_t
\mid\mathcal{G}_{t-1},\mathsf{Val}_t=0\bigr],
\end{multline*}
which forces the expectation to be zero. Since
$Z_t=(2c_i-1)\Lambda_t/B_\epsilon$ with $c_i$ and $B_\epsilon$ fixed,
the same conclusion holds for $Z_t$. Averaging over the observed
interval completes the proof.
\end{proof}
\section{Supplementary Experimental Analysis}
\label{app:ex}
\phantomsection
\label{app:oracle_attack}
\noindent\textbf{Key-aware oracle attack.}
The TM1 attacks in the main text have no secret key or detector feedback and mainly reduce the structural survival rate $\rho$. To probe whether a stronger key-aware attacker could also drive the directional channel negative, we construct a privileged greedy oracle outside TM1.

The oracle is strictly outside TM1 because it can access
secret-key-dependent verifier-side evidence. Given an edit budget $r$
and a candidate search width $k$, it greedily selects replacements that
most reduce the aligned evidence, directly targeting the directional
channel $\mu_{\mathrm{adv}}$. Table~\ref{tab:oracle_attack} reports the
resulting decomposition.

\begin{table}[!htbp]
\centering
\caption{Privileged key-aware oracle stress test outside TM1. The oracle uses secret-key-dependent verifier-side evidence to target the directional channel $\mu_{\mathrm{adv}}$.}
\label{tab:oracle_attack}
\setlength{\tabcolsep}{5pt}
\resizebox{\linewidth}{!}{%
\begin{tabular}{lccccc}
\toprule
Configuration
& $\hat{\rho}$
& $\hat{\mu}_c$
& $\hat{\mu}_{\mathrm{adv}}$
& $\hat{\Gamma}$
& Bit Acc. (\%) \\
\midrule
\multicolumn{6}{l}{\textbf{Key-aware greedy oracle}} \\
\quad $r=0.1,\ k=6$
& 0.788 & 0.602 & $-$0.094 & 0.454 & 95.16 \\
\quad $r=0.1,\ k=15$
& 0.778 & 0.600 & $-$0.156 & 0.432 & 93.19 \\
\quad $r=0.2,\ k=6$
& 0.627 & 0.531 & $-$0.098 & 0.296 & 79.03 \\
\quad $r=0.2,\ k=15$
& 0.613 & 0.527 & $-$0.163 & 0.260 & 74.18 \\
\quad $r=0.3,\ k=6$
& 0.392 & 0.455 & $-$0.092 & 0.122 & 67.38 \\
\quad $r=0.3,\ k=15$
& 0.381 & 0.449 & $-$0.161 & 0.071 & 60.47 \\
\quad $r=0.4,\ k=6$
& 0.241 & 0.398 & $-$0.102 & 0.019 & 54.65 \\
\quad $r=0.4,\ k=15$
& 0.244 & 0.392 & $-$0.159 & $-$0.025 & 47.35 \\
\bottomrule
\end{tabular}%
}
\end{table}

The oracle reduces recovery through both channels. Increasing the edit
budget lowers $\hat{\rho}$, while increasing the search width makes
$\hat{\mu}_{\mathrm{adv}}$ more negative. For example, at $r=0.1$ the
oracle already drives $\hat{\mu}_{\mathrm{adv}}$ to $-0.094$ with
$k=6$ and to $-0.156$ with $k=15$, unlike the TM1 attacks in the main
text whose adversarial drift remains close to zero. As $r$ and $k$
increase, the effective margin decreases monotonically: it falls to
$0.071$ at $r=0.3,k=15$ and becomes negative at $r=0.4,k=15$, where the
bit accuracy drops to $47.35\%$.

This stress test is not intended as a realistic TM1 attack. Instead, it
shows the boundary captured by our theory. Practical adaptive rewriting
attacks primarily attack the structural channel $\rho$, whereas a
key-aware oracle can directly attack the directional evidence channel
$\mu_{\mathrm{adv}}$. Once such an oracle can systematically induce
negative drift, the positive-margin condition in the robustness bound
can be weakened or even broken.

\noindent\textbf{Entropy-aware gating.}
\label{app:gating_ablation}
Figure~\ref{fig:gating_ablation} evaluates the adaptive gating mechanism. The full Adapt 1-2-3 strategy consistently outperforms fixed or restricted adaptive variants at all payload capacities, especially under editing attacks. This confirms that allocating more payload capacity to high-entropy regions is effective for robust extraction.

\begin{figure}[!htbp]
    \centering
    \includegraphics[width=\linewidth]{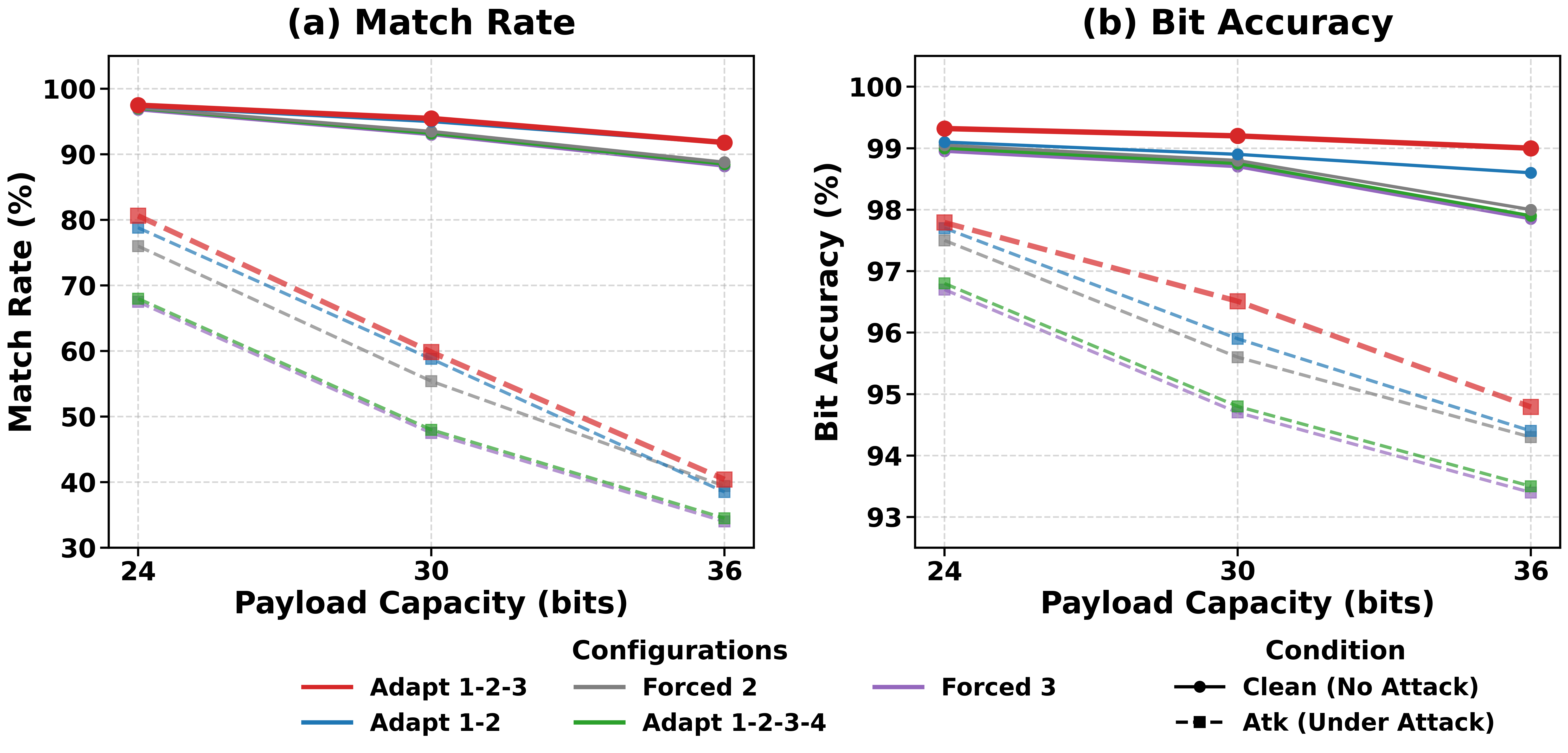}
    \caption{Comparison of entropy-aware gating strategies in clean and attacked settings.}
    \Description{Two panels compare match rate and bit accuracy for five gating strategies at 24, 30, and 36 bits. Adaptive 1-2-3 gating performs best in both clean and attacked settings, with its advantage most visible under attack.}
    \label{fig:gating_ablation}
\end{figure}

\FloatBarrier

\section{Kernels for Sector Evidence}
\label{app:kernels}

This appendix lists the symmetric finite-support kernels used for sector-evidence integration in Equation~\eqref{eq:sector_weight}. Recall
\[
w_t(s)=\int_{u\in I_t(x_t)} \mathcal K(u;A_{t,s}^{\mathrm{ver}})\,du.
\]

\noindent\textbf{Distance and support.}
We use the circular distance on $[0,1)$:
\begin{equation}
d(u,a)\triangleq \min\{|u-a|,\ 1-|u-a|\}\in\Bigl[0,\tfrac12\Bigr].
\end{equation}
At step $t$, when the arity is $a_t>0$, the circle is partitioned into $2^{a_t}$ sectors. We set the support radius of the kernel to half of the sector width:
\begin{equation}
h_t \triangleq 2^{-(a_t+1)}.
\end{equation}
All kernels below satisfy $\mathcal K(u;A_{t,s}^{\mathrm{ver}})=0$ whenever $d(u,A_{t,s}^{\mathrm{ver}})\ge h_t$.

\noindent\textbf{Definitions of finite-support kernels.}
Let $d_t \triangleq d(u,A_{t,s}^{\mathrm{ver}})$. We use
\begin{align}
\mathcal K_{\mathrm{box}}(u;A_{t,s}^{\mathrm{ver}})
&=\mathbb I(d_t<h_t),
\\
\mathcal K_{\mathrm{cos}}(u;A_{t,s}^{\mathrm{ver}})
&=\tfrac12\!\left(1+\cos\!\bigl(\pi d_t/h_t\bigr)\right)\,\mathbb I(d_t<h_t),
\\
\mathcal K_{\mathrm{tri}}(u;A_{t,s}^{\mathrm{ver}})
&=\bigl(1-d_t/h_t\bigr)_+,
\\
\mathcal K_{\mathrm{epa}}(u;A_{t,s}^{\mathrm{ver}})
&=\bigl(1-(d_t/h_t)^2\bigr)_+,
\\
\mathcal K_{\mathrm{biw}}(u;A_{t,s}^{\mathrm{ver}})
&=\bigl(1-(d_t/h_t)^2\bigr)_+^{2},
\end{align}
where $\mathbb I(\cdot)$ is the indicator function and $(x)_+\triangleq \max\{x,0\}$.

\FloatBarrier

\section{Algorithms}
This appendix gives detailed pseudocode for the PURA framework.
Algorithm~\ref{alg:embed} outlines the embedding procedure described
in Section~\ref{subsec:embedding}. It combines entropy-aware adaptive
gating with anchor-based inverse transform sampling.
Algorithm~\ref{alg:extract} presents the extraction and soft-decoding
pipeline described in Section~\ref{subsec:extraction}. It shows how
token probabilities are converted into bit-level log-likelihood ratios
for reliability-aware message recovery.

\begin{figure*}[t]
\SetAlgoNlRelativeSize{-1}
\begin{minipage}[t]{0.48\textwidth}
\makeatletter\@twocolumnfalse\makeatother
\begin{algorithm}[H]
\caption{PURA embedding at generation step $t$}
\label{alg:embed}
\footnotesize
\KwIn{Distribution $\pi_t$, context $x_{<t}$, key $\mathsf{sk}$,
payload $c$.}
\KwOut{Next token $x_t$.}

$\Pi \gets \text{DerivePermutation}(\mathsf{sk})$\;
$\tilde{\pi}_t \gets \text{ApplyPermutation}(\pi_t, \Pi)$\;
$H_t \gets \text{Entropy}(\tilde{\pi}_t)$\;
$h_{\mathrm{seed}} \gets \mathcal{H}_{\mathrm{seed}}(x_{t-W:t-1};\mathsf{sk})$\;
$a_t \gets \text{BitWidth}(H_t, h_{\mathrm{seed}}, \mathsf{Hist})$\;

\eIf{$a_t = 0$}{
    \KwRet{sample $x_t$ from $\pi_t$}\;
}{
    $\mathsf{Hist} \gets \mathsf{Hist} \cup \{h_{\mathrm{seed}}\}$\;
    $U_t \gets \mathrm{Unif}(h_{\mathrm{seed}})$\;
    $\mathrm{idx}(t) \gets
        \mathrm{Int}\!\bigl(
        \mathcal{H}_{\mathrm{idx}}(x_{t-W:t-1};\mathsf{sk})
        \bigr) \bmod |c|$\;
    $b_t \gets \{\,c_{(\mathrm{idx}(t)+j)\bmod|c|}\,\}_{j=0}^{a_t-1}$\;
    $S_t \gets \text{GrayInverse}(b_t)$\;
    $K_t \gets 2^{a_t}$\;
    $A_t \gets \bigl(U_t + (S_t + 0.5)/K_t\bigr)\bmod 1$\;
    $r \gets \text{InverseCDF}(\tilde{\pi}_t, A_t)$\;
    $x_t \gets \tilde{v}_r$\;
}
\KwRet{$x_t$}\;
\end{algorithm}
\end{minipage}\hfill
\begin{minipage}[t]{0.48\textwidth}
\makeatletter\@twocolumnfalse\makeatother

\begin{algorithm}[H]
\caption{PURA decoding}
\label{alg:extract}
\footnotesize
\KwIn{Observed sequence $\mathbf{x'}$, context $x_{<t}$,
key $\mathsf{sk}$, verifier model $\mathcal{M}_{\mathrm{ver}}$,
encoded length $|c|$, parity layout.}
\KwOut{Recovered payload $\hat{m}$.}
$\{\pi_t^{\mathrm{ver}}(\cdot\mid x_{<t})\}_{t=1}^{L}
    \gets \text{BatchForward}(\mathcal{M}_{\mathrm{ver}},\mathbf{x'})$\;
$\Pi \gets \text{DerivePermutation}(\mathsf{sk})$\;
$\Lambda_{\mathrm{total}}^{(i)} \gets 0$ for all $i\in\{0,\dots,|c|-1\}$\;
\For{$t \gets 1$ \KwTo $L$}{
    $\tilde{\pi}_t^{\mathrm{ver}} \gets
        \text{ApplyPermutation}(\pi_t^{\mathrm{ver}}, \Pi)$\;
    $H_t \gets \text{Entropy}(\tilde{\pi}_t^{\mathrm{ver}})$\;
    $h_{\mathrm{seed}} \gets
        \mathcal{H}_{\mathrm{seed}}(x_{t-W:t-1};\mathsf{sk})$\;
    $a_t \gets \text{BitWidth}(H_t,h_{\mathrm{seed}},\mathsf{Hist})$\;
    \If{$a_t \geq 1$}{
        $\mathsf{Hist} \gets \mathsf{Hist}\cup\{h_{\mathrm{seed}}\}$\;
        $U_t^{\mathrm{ver}} \gets \mathrm{Unif}(h_{\mathrm{seed}})$\;
        $\mathrm{idx}(t) \gets
            \mathrm{Int}(\mathcal{H}_{\mathrm{idx}}
            (x_{t-W:t-1};\mathsf{sk}))\bmod|c|$\;
        $\tilde{r}_t \gets \mathrm{rank}_{\Pi}(x_t)$\;
        $I_t(x_t) \gets
            \text{IntervalFromCDF}(\tilde{\pi}_t^{\mathrm{ver}},\tilde{r}_t)$\;
        $K_t \gets 2^{a_t}$\;
        \For{$s \gets 0$ \KwTo $K_t - 1$}{
            $A_{t,s}^{\mathrm{ver}} \gets
                (U_t^{\mathrm{ver}}+(s+0.5)/K_t)\bmod 1$\;
            $P_{t,s} \gets \displaystyle\int_{u\in I_t(x_t)}
                \mathcal{K}(u;\,A_{t,s}^{\mathrm{ver}})\,du$\;
        }
        $\{\Lambda_t^{(j)}\}_{j=0}^{a_t-1} \gets
            \text{SectorToBitLLRs}(\{P_{t,s}\}_{s=0}^{K_t-1})$\;
        \For{$j \gets 0$ \KwTo $a_t - 1$}{
            $i \gets (\mathrm{idx}(t)+j)\bmod|c|$\;
            $\Lambda_{\mathrm{total}}^{(i)} \gets
                \Lambda_{\mathrm{total}}^{(i)}+\Lambda_t^{(j)}$\;
        }
    }
}

$\hat{c}_i \gets
    \mathbb{I}(\Lambda_{\mathrm{total}}^{(i)}\ge0)$
for all $i\in\{0,\dots,|c|-1\}$\;
$\hat{m} \gets \text{ReliabilityAidedParity}\!\left(
    \hat{c},\,\{|\Lambda_{\mathrm{total}}^{(i)}|\},\,\text{parity layout}\right)$\;
\KwRet{$\hat{m}$}\;
\end{algorithm}
\end{minipage}
\Description{Two side-by-side pseudocode blocks summarize PURA embedding and decoding. Embedding derives the keyed permutation and gate, maps payload bits to a Gray-coded sector anchor, and samples a token. Decoding reconstructs token intervals, accumulates bit-level LLRs, and applies reliability-aided parity repair.}
\end{figure*}

\FloatBarrier
\end{document}